\documentclass{article}
 \usepackage{amssymb,nicefrac}
\usepackage{amsfonts}
\usepackage{graphicx}
\usepackage[fleqn]{amsmath}
\usepackage[varg]{txfonts}
\usepackage{stmaryrd}

\usepackage{framed}
\usepackage{color}
\usepackage[normalem]{ulem}
\usepackage{tikzfig}
\usepackage{placeins}

\usepackage{blkarray}

\usepackage{mathtools} 

\newenvironment{proof}{\textbf{Proof:}}{\hfill$\Box$\newline}

\tikzstyle{env}=[copoint,regular polygon rotate=0,minimum width=0.2cm, fill=black]

\tikzstyle{every picture}=[baseline=-0.25em]
\tikzstyle{dotpic}=[scale=0.5]
\tikzstyle{diredges}=[every to/.style={diredge}]
\tikzstyle{dot graph}=[shorten <=-0.1mm,shorten >=-0.1mm,scale=0.6]
\tikzstyle{plot point}=[circle,fill=black,minimum width=2mm,inner sep=0]

\tikzstyle{braceedge}=[decorate,decoration={brace,amplitude=2mm,raise=-1mm}]
\tikzstyle{small braceedge}=[decorate,decoration={brace,amplitude=1mm,raise=-1mm}]
\tikzstyle{left hook arrow}=[left hook-latex]
\tikzstyle{right hook arrow}=[right hook-latex]

\tikzstyle{dtriangle}=[fill=yellow,draw=black,shape=isosceles triangle,shape border rotate=-90,isosceles triangle stretches=true,inner sep=0.8pt,minimum width=0.25cm,minimum height=2mm]
\tikzstyle{vtriang}=[fill=yellow,draw=black,shape=isosceles triangle,shape border rotate=180,isosceles triangle stretches=true,inner sep=0.8pt,minimum width=0.25cm,minimum height=2mm]
\tikzstyle{vrt}=[fill=yellow,draw=black,shape=isosceles triangle,shape border rotate=0,isosceles triangle stretches=true,inner sep=0.8pt,minimum width=0.25cm,minimum height=2mm]
\tikzstyle{H box}=[rectangle,fill=yellow,draw=black,xscale=0.8,yscale=0.8, inner sep=0.6pt]
\tikzstyle{gbox}=[rectangle,fill=green,draw=black,xscale=1.0,yscale=1.0, inner sep=1.pt]
\tikzstyle{rbox}=[rectangle,fill=red,draw=black,xscale=1.0,yscale=1.0, inner sep=1.pt]
\tikzstyle{triangle}=[fill=yellow,draw=black,shape=isosceles triangle,shape border rotate=90,isosceles triangle stretches=true,inner sep=0.8pt,minimum width=0.25cm,minimum height=2mm]

\tikzstyle{bn}=[circle,fill=black,draw=black,scale=.8]
\tikzstyle{wn}=[circle,fill=white,draw=black,scale=.6]
\tikzstyle{dn}=[circle,fill=none,draw=gray]
\tikzstyle{bspider}=[fill=black,draw=black,scale=1,shape=isosceles triangle,shape border rotate=-90,isosceles triangle stretches=true,inner sep=1pt,minimum width=0.4cm,minimum height=3mm]
\tikzstyle{dbspider}=[fill=black,draw=black,scale=1,shape=isosceles triangle,shape border rotate=90,isosceles triangle stretches=true,inner sep=1pt,minimum width=0.4cm,minimum height=3mm]
\tikzstyle{L}=[rectangle,shape=rectangle,fill=green,draw=black]

\tikzstyle{Z dot}=[inner sep=0mm, minimum size=2mm, shape=circle, draw=black, fill={rgb,255: red,221; green,255; blue,221}]
\tikzstyle{Z phase dot}=[minimum size=5mm, font={\footnotesize\boldmath}, shape=rectangle, rounded corners=2mm, inner sep=0.2mm, outer sep=-2mm, scale=0.8, draw=black, fill={rgb,255: red,221; green,255; blue,221}]
\tikzstyle{X dot}=[Z dot, shape=circle, draw=black, fill={rgb,255: red,255; green,136; blue,136}]
\tikzstyle{X phase dot}=[Z phase dot, fill={rgb,255: red,255; green,136; blue,136}, font={\footnotesize\boldmath}]
\tikzstyle{hadamard edge}=[-, dashed, dash pattern=on 2pt off 0.5pt, thick, draw={rgb,255: red,68; green,136; blue,255}]

\tikzstyle{black dot}=[inner sep=0.7mm,minimum width=0pt,minimum height=0pt,fill=black,draw=black,shape=circle]

\tikzstyle{dot}=[black dot]
\tikzstyle{smalldot}=[inner sep=0.4mm,minimum width=0pt,minimum height=0pt,fill=black,draw=black,shape=circle]
\tikzstyle{white dot}=[dot,fill=white]
\tikzstyle{antipode}=[white dot,inner sep=0.3mm,font=\footnotesize]
\tikzstyle{smallwhitedot}=[smalldot,fill=white]
\tikzstyle{alt white dot}=[white dot,label={[xshift=3.07mm,yshift=-0.05mm,font=\footnotesize]left:$*$}]
\tikzstyle{gray dot}=[dot,fill=gray!40!white]
\tikzstyle{smallgraydot}=[smalldot,fill=gray!40!white]
\tikzstyle{box vertex}=[draw=black,rectangle]
\tikzstyle{small box}=[box vertex,fill=white]
\tikzstyle{whitebg}=[fill=white,inner sep=2pt]
\tikzstyle{graph state vertex}=[sg vertex,fill=black]

\tikzstyle{wide copoint}=[fill=white,draw=black,shape=isosceles triangle,shape border rotate=90,isosceles triangle stretches=true,inner sep=1pt,minimum width=1.5cm,minimum height=5mm]
\tikzstyle{wide point}=[fill=white,draw=black,shape=isosceles triangle,shape border rotate=-90,isosceles triangle stretches=true,inner sep=1pt,minimum width=1.5cm,minimum height=4mm]
\tikzstyle{very wide copoint}=[fill=white,draw=black,shape=isosceles triangle,shape border rotate=-90,isosceles triangle stretches=true,inner sep=1pt,minimum width=2.5cm,minimum height=4mm]
\tikzstyle{very wide empty copoint}=[draw=black,shape=isosceles triangle,shape border rotate=-90,isosceles triangle stretches=true,inner sep=1pt,minimum width=2.5cm,minimum height=4mm]
\tikzstyle{symm}=[ultra thick,shorten <=-1mm,shorten >=-1mm]

\tikzstyle{square box}=[rectangle,fill=white,draw=black,minimum height=5mm,minimum width=5mm,font=\small]
\tikzstyle{square gray box}=[rectangle,fill=gray!30,draw=black,minimum height=6mm,minimum width=6mm]
\tikzstyle{copoint}=[regular polygon,regular polygon sides=3,draw=black,scale=0.75,inner sep=-0.5pt,minimum width=7mm,fill=white]
\tikzstyle{point}=[regular polygon,regular polygon sides=3,draw=black,scale=0.75,inner sep=-0.5pt,minimum width=7mm,fill=white,regular polygon rotate=180]
\tikzstyle{gray point}=[point,fill=gray!40!white]
\tikzstyle{gray copoint}=[copoint,fill=gray!40!white]

\newcommand{\edgearrow}{{\arrow[black]{>}}}
\newcommand{\edgetick}{{\arrow[black,scale=0.7,very thick]{|}}}

\tikzstyle{diredge}=[->]
\tikzstyle{rdiredge}=[<-]
\tikzstyle{medium diredge}=[->]

\tikzstyle{short diredge}=[->]
\tikzstyle{halfedge}=[-)]
\tikzstyle{other halfedge}=[(-]
\tikzstyle{freeedge}=[(-)]
\tikzstyle{white edge}=[line width=5pt,white]
\tikzstyle{tick}=[postaction=decorate,decoration={markings, mark=at position 0.5 with \edgetick}]
\tikzstyle{small map edge}=[|-latex, gray!60!blue, shorten <=0.9mm, shorten >=0.5mm]
\tikzstyle{thick dashed edge}=[very thick,dashed,gray!40]
\tikzstyle{map edge}=[|-latex,very thick, gray!40, shorten <=1mm, shorten >=0.5mm]
\tikzstyle{tickedge}=[postaction=decorate,
  decoration={markings, mark=at position 0.5 with \edgetick}]
\tikzstyle{dirtickedge}=[postaction=decorate,
  decoration={markings, mark=at position 0.5 with \edgetick},
  decoration={markings, mark=at position 0.85 with \edgearrow}]
\tikzstyle{dirdoubletickedge}=[postaction=decorate,
  decoration={markings, mark=at position 0.4 with \edgetick},
  decoration={markings, mark=at position 0.6 with \edgetick},
  decoration={markings, mark=at position 0.85 with \edgearrow}]

\makeatletter
\newcommand{\boxshape}[3]{%
\pgfdeclareshape{#1}{
\inheritsavedanchors[from=rectangle] 
\inheritanchorborder[from=rectangle]
\inheritanchor[from=rectangle]{center}
\inheritanchor[from=rectangle]{north}
\inheritanchor[from=rectangle]{south}
\inheritanchor[from=rectangle]{west}
\inheritanchor[from=rectangle]{east}
\backgroundpath{
\southwest \pgf@xa=\pgf@x \pgf@ya=\pgf@y
\northeast \pgf@xb=\pgf@x \pgf@yb=\pgf@y

\@tempdima=#2
\@tempdimb=#3

\pgfpathmoveto{\pgfpoint{\pgf@xa - 5pt + \@tempdima}{\pgf@ya}}
\pgfpathlineto{\pgfpoint{\pgf@xa - 5pt - \@tempdima}{\pgf@yb}}
\pgfpathlineto{\pgfpoint{\pgf@xb + 5pt + \@tempdimb}{\pgf@yb}}
\pgfpathlineto{\pgfpoint{\pgf@xb + 5pt - \@tempdimb}{\pgf@ya}}
\pgfpathlineto{\pgfpoint{\pgf@xa - 5pt + \@tempdima}{\pgf@ya}}
\pgfpathclose
}
}}

\boxshape{NEbox}{0pt}{8pt}
\boxshape{SEbox}{0pt}{-8pt}
\boxshape{NWbox}{8pt}{0pt}
\boxshape{SWbox}{-8pt}{0pt}
\makeatother

\tikzstyle{map}=[draw,shape=NEbox,inner sep=7pt]
\tikzstyle{mapdag}=[draw,shape=SEbox,inner sep=7pt]
\tikzstyle{maptrans}=[draw,shape=SWbox,inner sep=7pt]
\tikzstyle{mapconj}=[draw,shape=NWbox,inner sep=7pt]

\tikzstyle{probs}=[shape=semicircle,fill=gray!40!white,draw=black,shape border rotate=180,minimum width=1.2cm]

\tikzstyle{arrs}=[-latex,font=\small,auto]
\tikzstyle{arrow plain}=[arrs]
\tikzstyle{arrow dashed}=[dashed,arrs]
\tikzstyle{arrow bold}=[very thick,arrs]
\tikzstyle{arrow hide}=[draw=white!0,-]
\tikzstyle{arrow reverse}=[latex-]
\tikzstyle{cdnode}=[]

\tikzstyle{gn}=[dot,fill=green,minimum width=0.25cm,inner sep=0pt]
\tikzstyle{rn}=[dot,fill=red,inner sep=0pt,minimum width=0.25cm]

\tikzstyle{rc}=[dot,thick,fill=white,draw = red,minimum width=0.3cm,inner sep=0pt]
\tikzstyle{gc}=[dot,thick,fill=white,draw= green,inner sep=0pt,minimum width=0.3cm]
\tikzstyle{bc}=[dot,thick,fill=white,draw= blue,minimum width=0.3cm]

\tikzstyle{label}=[circle,fill=white,minimum width=0.3cm]

\tikzstyle{clocklabel}=[dot,fill=yellow,draw=black,font=\tiny,inner sep=0.75pt]

\tikzstyle{rsn}=[circle split,draw,fill=red,font=\tiny,inner sep=0.75pt]
\tikzstyle{gsn}=[circle split,draw,fill=green,font=\tiny,inner sep=0.75pt]
\tikzstyle{bsn}=[circle split,draw,fill=blue,font=\tiny,inner sep=0.75pt]

\tikzstyle{rsc}=[circle split,thick,draw= red,draw,fill=white,font=\tiny,inner sep=0.75pt]
\tikzstyle{gsc}=[circle split,thick,draw= green,draw,fill=white,font=\tiny,inner sep=0.75pt]
\tikzstyle{bsc}=[circle split,thick,draw= blue,draw,fill=white,font=\tiny,inner sep=0.75pt]

\tikzstyle{cnot}=[fill=white,shape=circle,inner sep=-1.4pt]
\tikzstyle{wire label}=[font=\tiny, auto]

\tikzstyle{square box2}=[rectangle,fill=white,draw=black,minimum height=5mm,minimum width=10mm,font=\small]

\tikzstyle{scalar}=[rectangle,shape=diamond,fill=white,draw=black, inner sep=0.8pt]

\newcommand{\bra}[1]{\ensuremath{\left\langle #1 \right|}}
\newcommand{\ket}[1]{\ensuremath{\left|  #1 \right\rangle}}

\tikzstyle{cdiag}=[matrix of math nodes, row sep=3em, column sep=3em, text height=1.5ex, text depth=0.25ex,inner sep=0.5em]
\tikzstyle{arrow above}=[transform canvas={yshift=0.5ex}]
\tikzstyle{arrow below}=[transform canvas={yshift=-0.5ex}]

\newtheorem{Th}{Theorem}[section]
\newtheorem{theorem}[Th]{Theorem}
 
\newtheorem{lemma}[Th]{Lemma}

\newtheorem{remark}[Th]{Remark}

\makeatletter
\newcommand{\vast}{\bBigg@{6.5}}
\makeatother

\newcommand{\vertrule}[1][1ex]{\rule{.4pt}{#1}}

\newcommand{\rddots}{\rotatebox{90}{$\ddots$}}

\title{ A non-anyonic qudit ZW-calculus }
\author{Quanlong Wang\\
Cambridge Quantum Computing Ltd. \\ harny.wang@cambridgequantum.com}
\begin{document}
\date{}\maketitle
\begin{abstract}
ZW-calculus is a useful graphical language for pure qubit quantum computing. It is via the translation of the completeness of ZW-calculus that the first proof of completeness of ZX-calculus was obtained. A d-level generalisation of qubit  ZW-calculus (anyonic qudit ZW-calculus) has been given in \cite{amar} which is universal for pure qudit quantum computing. However, the interpretation of the W spider in this type of ZW-calculus has so-called  q-binomial coefficients involved, thus makes computation quite complicated. In this paper, we give a new type of qudit ZW-calculus which has generators and rewriting rules similar to that of the qubit ZW-calculus. Especially, the Z spider is exactly the same as that of the qudit ZX-calculus as given in \cite{wangqufinite}, and the new W spider has much simpler interpretation as a linear map.  Furthermore, we establish a translation between this qudit ZW-calculus and  the qudit ZX-calculus which is universal as shown in \cite{wangqufinite}, therefore this qudit ZW-calculus is also universal for pure qudit quantum computing.
\end{abstract}

\section{Introduction}
As is well known, everything in the world can be seen as a process \cite{Whitehead1929}. But do we have a mathematical framework for general processes? Fortunately yes, symmetric monoidal category theory (a.k.a. process theory ) is such a formalism. In process theory, each process is a change from something typed as $A$ to something typed as $B$, which is denoted by a diagram read from top to bottom: 

$$%
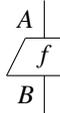
\beginpgfgraphicnamed{TikZit//singleprocess2}
\InputIfFileExists{TikZit//singleprocess2.tikz}{}{\input{./figures/TikZit//singleprocess2.tikz}}
\endpgfgraphicnamed $$

Two processes that happened  sequentially are denoted as

 $$%
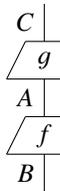
\beginpgfgraphicnamed{TikZit//sequentialproc2}
\InputIfFileExists{TikZit//sequentialproc2.tikz}{}{\input{./figures/TikZit//sequentialproc2.tikz}}
\endpgfgraphicnamed $$
  
 While simultaneously happened processes are depicted as
 
 $$ %
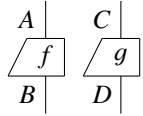
\beginpgfgraphicnamed{TikZit//parallelpro2}
\InputIfFileExists{TikZit//parallelpro2.tikz}{}{\input{./figures/TikZit//parallelpro2.tikz}}
\endpgfgraphicnamed $$
 
Therefore any two processes are either in one same sequential train of boxes or in two parallel sequential train of boxes respectively.  Note that we have put  the two simultaneously happened processes in parallel in a plane, so one has to be on the left and the other has to be  on the right. Then the symmetry process $%
\beginpgfgraphicnamed{TikZit//swap2v}
\InputIfFileExists{TikZit//swap2v.tikz}{}{\input{./figures/TikZit//swap2v.tikz}}
\endpgfgraphicnamed$ in process theory comes to the rescue to break this artificial asymmetry in a way that  if we swap the positions of the two boxes, they should be essentially the same while all the types still match:
  $$ %
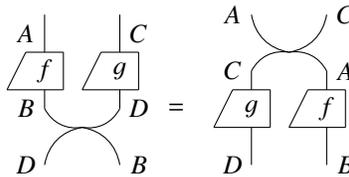
\beginpgfgraphicnamed{TikZit//swap2process2}
\InputIfFileExists{TikZit//swap2process2.tikz}{}{\input{./figures/TikZit//swap2process2.tikz}}
\endpgfgraphicnamed $$
  
In addition to the left-right symmetry, we would also like to have an up-down symmetry:  when we turn a diagram upside-down, then essentially the property of the corresponding process won't be changed if the types are kept matching. This can be realised by a self-dual compact structure \cite{Coeckebk} in the following way: 
$$ %
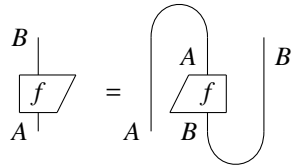
\beginpgfgraphicnamed{TikZit//topdownsymmetry}
\InputIfFileExists{TikZit//topdownsymmetry.tikz}{}{\input{./figures/TikZit//topdownsymmetry.tikz}}
\endpgfgraphicnamed $$
 
where a self-dual compact structure compatible with the symmetry means 

 \[
\beginpgfgraphicnamed{TikZit//compactstructure_capdit}
\InputIfFileExists{TikZit//compactstructure_capdit.tikz}{}{\input{./figures/TikZit//compactstructure_capdit.tikz}}
\endpgfgraphicnamed \qquad\quad %
\beginpgfgraphicnamed{TikZit//compactstructure_cupdit}
\InputIfFileExists{TikZit//compactstructure_cupdit.tikz}{}{\input{./figures/TikZit//compactstructure_cupdit.tikz}}
\endpgfgraphicnamed \]
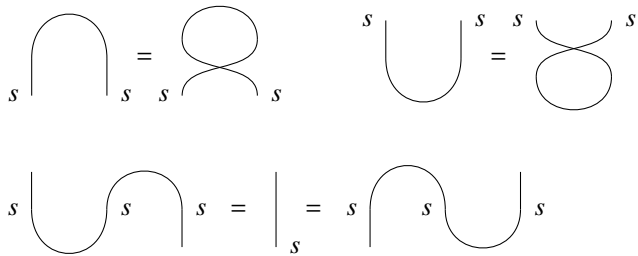
\[ %
\beginpgfgraphicnamed{TikZit//compactstructure_snakedit}
\InputIfFileExists{TikZit//compactstructure_snakedit.tikz}{}{\input{./figures/TikZit//compactstructure_snakedit.tikz}}
\endpgfgraphicnamed\]

To summarise, we now obtained a compact closed category which could characterise the interaction of general processes and composes the backbone of categorical quantum mechanics \cite{Coeckesamson}. However,  if we want to apply process theory to  a particular problem like quantum circuit optimisation or quantum entanglement classification, then quite probably we may need a fine-grained version of process theory where all the boxes are explicitly detailed, two typical examples of such a version are ZX-calculus \cite{CoeckeDuncan} and ZW-calculus \cite{amar1, amarngwanglics}. 

ZW-calculus is a graphical language which characterises relations between the GHZ and W algebras \cite{Coeckealeksmen} while using undirected diagrams. It has been shown that ZW-calculus is a very useful tool in dealing with the multipartite qubit entanglement classification problem  \cite{Coeckealeksmen, amar}. However, when it is generalised to an arbitrary dimension $d$ in the anyonic line as a standard generalisation of the fermionic line \cite{amar}, the resulted qudit ZW-calculus has complicated interpretation for its $W$ node and white phase node:
 $$\left\llbracket %
\beginpgfgraphicnamed{TikZit/blacktriangledit}
\begin{tikzpicture}
	\begin{pgfonlayer}{nodelayer}
		\node [style=none] (0) at (0.25, -0.5) {};
		\node [style=none] (1) at (0, 0.5) {};
		\node [style=none] (2) at (-0.25, -0.5) {};
		\node [style=dbspider] (3) at (0, 0) {};
	\end{pgfonlayer}
	\begin{pgfonlayer}{edgelayer}
		\draw (1.center) to (3);
		\draw (3) to (2.center);
		\draw (3) to (0.center);
	\end{pgfonlayer}
\end{tikzpicture}}
\endpgfgraphicnamed \right\rrbracket= \sum_{n=0}^{d-1}\sum_{k=0}^{n} \sqrt{\binom{n}{k}_q}\ket{k} \ket{n-k} \bra{n},  \left\llbracket %
\beginpgfgraphicnamed{TikZit/whitephaselambda}
\InputIfFileExists{TikZit/whitephaselambda.tikz}{}{\input{./figures/TikZit/whitephaselambda.tikz}}
\endpgfgraphicnamed \right\rrbracket=  \sum_{n=0}^{d-1}\lambda^n\sqrt{[n]_q !} \ket{n} \ket{n}\bra{n},$$
where
 \begin{equation*}
q = e^{2 \pi i/d}, ~~ [j]_q = \sum_{k=0}^{j-1} q^k, ~~\binom{n}{k}_q = \frac{[n]_q!}{[k]_q ! [n-k]_q !}.
\end{equation*}

To avoid complicated coefficients like q-binomial numbers, in this paper we propose a non-anyionic qudit ZW-calculus for arbitrary dimension $d$. The main feature of this non-anyionic qudit ZW-calculus lies in that the non-zero coefficients of the interpretation of the $W$ node are just $1$, and the coefficients of the interpretation of the white phase node are independent arbitrary complex numbers (the first coefficient is fixed to be $1$) which is the same as that for the $Z$ spider as given in algebraic qudit ZX-calculus  \cite{wangqufinite}. The other generators of the non-anyionic qudit ZW-calculus are the same as deployed in the anyionic qudit ZW-calculus  \cite{amar}.  Furthermore, we obtain a set of rewriting rules and their derivations which can be seen as a direct generalisation of rules for fermionic qubit ZW-calculus \cite{amarngwanglics}. As a consequence, we get a bilateral translation between the non-anyionic qudit ZW-calculus and algebraic qudit ZX-calculus, therefore the universality of the former is inherited from the latter. Finally, we unify all the finite dimensional non-anyionic  ZW-calculi into a single framework called qufinite ZW-calculus in the same way as that for the qufinite ZX-calculus.


\section{ZW-calculus}
ZW-calculus is a PROP which could be represented by generators and rewriting rules  \cite{BaezCR}. In this section we give the generators and rewriting rules for qudit ZW-calculus based on any dimension $d \geq 2$.

\FloatBarrier
 \begin{table}
\begin{center}
\begin{tabular}{|r@{~}r@{~}c@{~}c|r@{~}r@{~}c@{~}c|}
\hline
$R_{Z,\overrightarrow{a}}^{(n,m)}$&$:$&$n\to m$ & %
\beginpgfgraphicnamed{TikZit/whitespiderdit}
\InputIfFileExists{TikZit/whitespiderdit.tikz}{}{\input{./figures/TikZit/whitespiderdit.tikz}}
\endpgfgraphicnamed  & W&$:$ &$1\to 2$ &%
\beginpgfgraphicnamed{TikZit/blacktriangledit}
\InputIfFileExists{TikZit/blacktriangledit.tikz}{}{\input{./figures/TikZit/blacktriangledit.tikz}}
\endpgfgraphicnamed \\ \hline
   $\tau$&$:$&$2\to 2$&%
\beginpgfgraphicnamed{TikZit/braid}
\InputIfFileExists{TikZit/braid.tikz}{}{\input{./figures/TikZit/braid.tikz}}
\endpgfgraphicnamed &   $\tau^{-1}$&$:$&$ 2\to 2$& %
\beginpgfgraphicnamed{TikZit/braidd2}
\InputIfFileExists{TikZit/braidd2.tikz}{}{\input{./figures/TikZit/braidd2.tikz}}
\endpgfgraphicnamed\\ \hline

   $\mathbb I$&$:$&$1\to 1$&%
\beginpgfgraphicnamed{TikZit/Id}
\InputIfFileExists{TikZit/Id.tikz}{}{\input{./figures/TikZit/Id.tikz}}
\endpgfgraphicnamed &   $\sigma$&$:$&$ 2\to 2$& %
\beginpgfgraphicnamed{TikZit/swap}
\InputIfFileExists{TikZit/swap.tikz}{}{\input{./figures/TikZit/swap.tikz}}
\endpgfgraphicnamed\\ \hline
   $C_a$&$:$&$ 0\to 2$& %
\beginpgfgraphicnamed{TikZit/cap}
\InputIfFileExists{TikZit/cap.tikz}{}{\input{./figures/TikZit/cap.tikz}}
\endpgfgraphicnamed &$ C_u$&$:$&$ 2\to 0$&%
\beginpgfgraphicnamed{TikZit/cup}
\InputIfFileExists{TikZit/cup.tikz}{}{\input{./figures/TikZit/cup.tikz}}
\endpgfgraphicnamed \\\hline
\end{tabular}\caption{ Generators of qudit ZW-calculus, where $m,n\in \mathbb N$, $\protect\overrightarrow{a}=(a_1, \cdots, a_{d-1}), a_i \in \mathbb C. $} \label{qbzxgeneratordit}
\end{center}
\end{table}
\FloatBarrier

We  define the empty diagram and some other diagrams in terms of the generators as follows: 
 \begin{center}\[
\beginpgfgraphicnamed{TikZit/zwditshortnotecodot}
\InputIfFileExists{TikZit/zwditshortnotecodot.tikz}{}{\input{./figures/TikZit/zwditshortnotecodot.tikz}}
\endpgfgraphicnamed \quad\quad\quad\quad %
\beginpgfgraphicnamed{TikZit/downblacktriangle}
\InputIfFileExists{TikZit/downblacktriangle.tikz}{}{\input{./figures/TikZit/downblacktriangle.tikz}}
\endpgfgraphicnamed\quad\quad\quad\quad  %
\beginpgfgraphicnamed{TikZit/unitdowanblack}
\InputIfFileExists{TikZit/unitdowanblack.tikz}{}{\input{./figures/TikZit/unitdowanblack.tikz}}
\endpgfgraphicnamed 
\]\end{center} 
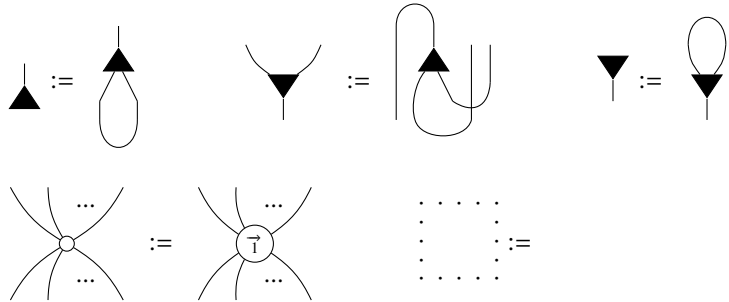
\[   %
\beginpgfgraphicnamed{TikZit/whitenophasedef}
\InputIfFileExists{TikZit/whitenophasedef.tikz}{}{\input{./figures/TikZit/whitenophasedef.tikz}}
\endpgfgraphicnamed \quad\quad\quad\quad   %
\beginpgfgraphicnamed{TikZit/emptysquare}
\InputIfFileExists{TikZit/emptysquare.tikz}{}{\input{./figures/TikZit/emptysquare.tikz}}
\endpgfgraphicnamed:= \]
where $\overrightarrow{1}=\overbrace{(1,\cdots,1)}^{d-1}$.

 The diagrams of the qudit ZX-calculus have the following standard interpretation $ \left\llbracket \cdot \right\rrbracket$:

 \[
 \left\llbracket %
\beginpgfgraphicnamed{TikZit/whitespiderdit}
\InputIfFileExists{TikZit/whitespiderdit.tikz}{}{\input{./figures/TikZit/whitespiderdit.tikz}}
\endpgfgraphicnamed \right\rrbracket=\sum_{j=0}^{d-1}a_j\ket{j}^{\otimes m}\bra{j}^{\otimes n}, a_0=1,  \protect\overrightarrow{a}=(a_1, \cdots, a_{d-1}), a_i \in \mathbb C,
\]

 \[
\left\llbracket %
\beginpgfgraphicnamed{TikZit/blacktriangledit}
\begin{tikzpicture}
	\begin{pgfonlayer}{nodelayer}
		\node [style=none] (0) at (0.25, -0.5) {};
		\node [style=none] (1) at (0, 0.5) {};
		\node [style=none] (2) at (-0.25, -0.5) {};
		\node [style=dbspider] (3) at (0, 0) {};
	\end{pgfonlayer}
	\begin{pgfonlayer}{edgelayer}
		\draw (1.center) to (3);
		\draw (3) to (2.center);
		\draw (3) to (0.center);
	\end{pgfonlayer}
\end{tikzpicture}}
\endpgfgraphicnamed \right\rrbracket=\ket{00}\bra{0}+\sum_{i=1}^{d-1}(\ket{0i}+\ket{i0})\bra{i}, \quad \left\llbracket%
\beginpgfgraphicnamed{TikZit/blackcodot}
\begin{tikzpicture}
	\begin{pgfonlayer}{nodelayer}
		\node [style=none] (0) at (0, 0.25) {};
		\node [style=dbspider] (1) at (0, -0.25) {};
	\end{pgfonlayer}
	\begin{pgfonlayer}{edgelayer}
		\draw (0.center) to (1);
	\end{pgfonlayer}
\end{tikzpicture}}
\endpgfgraphicnamed\right\rrbracket=\bra{0},  \]

 \[
\left\llbracket %
\beginpgfgraphicnamed{TikZit/blackcotriangledit}
\begin{tikzpicture}
	\begin{pgfonlayer}{nodelayer}
		\node [style=none] (0) at (0, -0.5) {};
		\node [style=none] (1) at (0.25, 0.5) {};
		\node [style=none] (2) at (-0.25, 0.5) {};
		\node [style=bspider] (3) at (0, 0) {};
	\end{pgfonlayer}
	\begin{pgfonlayer}{edgelayer}
		\draw (3) to (0.center);
		\draw (3) to (2.center);
		\draw (3) to (1.center);
	\end{pgfonlayer}
\end{tikzpicture}}
\endpgfgraphicnamed \right\rrbracket=\ket{0}\bra{00}+\sum_{i=1}^{d-1}\ket{i}(\bra{0i}+\bra{i0}), \quad \left\llbracket%
\beginpgfgraphicnamed{TikZit/blackdot}
\begin{tikzpicture}
	\begin{pgfonlayer}{nodelayer}
		\node [style=none] (0) at (0, -0.25) {};
		\node [style=bspider] (1) at (0, 0.25) {};
	\end{pgfonlayer}
	\begin{pgfonlayer}{edgelayer}
		\draw (1) to (0.center);
	\end{pgfonlayer}
\end{tikzpicture}}
\endpgfgraphicnamed\right\rrbracket=\ket{0},  \quad \left\llbracket%
\beginpgfgraphicnamed{TikZit/emptysquare}
\InputIfFileExists{TikZit/emptysquare.tikz}{}{\input{./figures/TikZit/emptysquare.tikz}}
\endpgfgraphicnamed\right\rrbracket=1,\]

\[
\left\llbracket%
\beginpgfgraphicnamed{TikZit/braid}
\InputIfFileExists{TikZit/braid.tikz}{}{\input{./figures/TikZit/braid.tikz}}
\endpgfgraphicnamed\right\rrbracket=\sum_{ j,k=0}^{d-1}\xi^{jk}\ket{jk}\bra{kj}, \xi=e^{i\frac{2\pi}{d}}, \quad
\left\llbracket%
\beginpgfgraphicnamed{TikZit/braidd2}
\InputIfFileExists{TikZit/braidd2.tikz}{}{\input{./figures/TikZit/braidd2.tikz}}
\endpgfgraphicnamed\right\rrbracket=\sum_{ j,k=0}^{d-1}\bar{\xi}^{jk}\ket{jk}\bra{kj}, \bar{\xi}=e^{-i\frac{2\pi}{d}},  \]

\[
   \left\llbracket%
\beginpgfgraphicnamed{TikZit/Id}
\begin{tikzpicture}
	\begin{pgfonlayer}{nodelayer}
		\node [style=none] (1) at (0.5, 0.3) {};
		\node [style=none] (2) at (0.5, -0.3) {};
		\node [style=none] (3) at (0.5, -0.5) {};
		\node [style=none] (4) at (0.5, 0.5) {};
	\end{pgfonlayer}
	\begin{pgfonlayer}{edgelayer}
		\draw (1.center) to (2.center);
	\end{pgfonlayer}
\end{tikzpicture}}
\endpgfgraphicnamed\right\rrbracket=I_d=\sum_{j=0}^{d-1}\ket{j}\bra{j},  \quad
 \left\llbracket%
\beginpgfgraphicnamed{TikZit/swap}
\InputIfFileExists{TikZit/swap.tikz}{}{\input{./figures/TikZit/swap.tikz}}
\endpgfgraphicnamed\right\rrbracket=\sum_{i, j=0}^{d-1}\ket{ji}\bra{ij},\quad
  \left\llbracket%
\beginpgfgraphicnamed{TikZit/cap}
\begin{tikzpicture}
	\begin{pgfonlayer}{nodelayer}
		\node [style=none] (0) at (-0.5, -0.25) {};
		\node [style=none] (1) at (0.5, -0.25) {};
	\end{pgfonlayer}
	\begin{pgfonlayer}{edgelayer}
		\draw [bend left=90, looseness=1.50] (0.center) to (1.center);
	\end{pgfonlayer}
\end{tikzpicture}}
\endpgfgraphicnamed\right\rrbracket=\sum_{j=0}^{d-1}\ket{jj}, \quad
   \left\llbracket%
\beginpgfgraphicnamed{TikZit/cup}
\begin{tikzpicture}
	\begin{pgfonlayer}{nodelayer}
		\node [style=none] (0) at (-0.5, 0.25) {};
		\node [style=none] (1) at (0.5, 0.25) {};
	\end{pgfonlayer}
	\begin{pgfonlayer}{edgelayer}
		\draw [bend right=90, looseness=1.50] (0.center) to (1.center);
	\end{pgfonlayer}
\end{tikzpicture}}
\endpgfgraphicnamed\right\rrbracket=\sum_{j=0}^{d-1}\bra{jj},   
      \]

\[  \llbracket D_1\otimes D_2  \rrbracket =  \llbracket D_1  \rrbracket \otimes  \llbracket  D_2  \rrbracket, \quad 
 \llbracket D_1\circ D_2  \rrbracket =  \llbracket D_1  \rrbracket \circ  \llbracket  D_2  \rrbracket.
  \]
  
  \begin{remark}
Except for the white spider and the black node $W$, all the other generators including the two braidings $\tau$ and $\tau^{-1}$ are the same as those for the  anyonic qudit  ZW-calculus [1].  When $d=2$, the black node %
\beginpgfgraphicnamed{TikZit/blacktriangledit}
\begin{tikzpicture}
	\begin{pgfonlayer}{nodelayer}
		\node [style=none] (0) at (0.25, -0.5) {};
		\node [style=none] (1) at (0, 0.5) {};
		\node [style=none] (2) at (-0.25, -0.5) {};
		\node [style=dbspider] (3) at (0, 0) {};
	\end{pgfonlayer}
	\begin{pgfonlayer}{edgelayer}
		\draw (1.center) to (3);
		\draw (3) to (2.center);
		\draw (3) to (0.center);
	\end{pgfonlayer}
\end{tikzpicture}}
\endpgfgraphicnamed is just the same as the original W node %
\beginpgfgraphicnamed{TikZit/oldblacknoded2}
\InputIfFileExists{TikZit/oldblacknoded2.tikz}{}{\input{./figures/TikZit/oldblacknoded2.tikz}}
\endpgfgraphicnamed as given in \cite{amar1} and \cite{amar}.
 \end{remark}

Now we give the rules of non-anyonic  qudit ZW-calculus.
   \begin{figure}[!h]
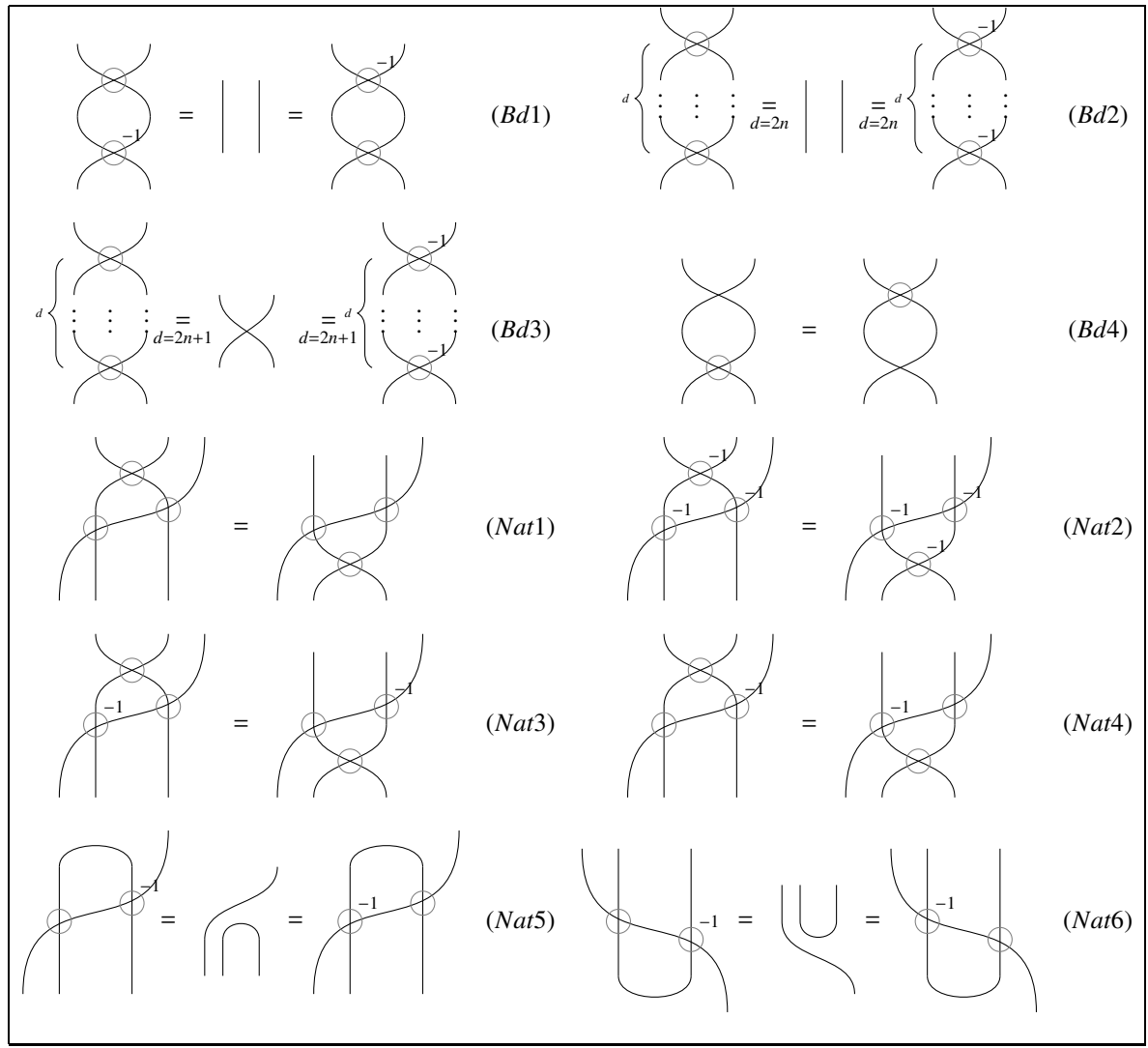

\begin{center} 
\[
\quad \qquad\begin{array}{|cccc|}
\hline
\beginpgfgraphicnamed{TikZit/braidscompose2}
\InputIfFileExists{TikZit/braidscompose2.tikz}{}{\input{./figures/TikZit/braidscompose2.tikz}}
\endpgfgraphicnamed&(Bd1) &%
\beginpgfgraphicnamed{TikZit/dbraidscompose}
\InputIfFileExists{TikZit/dbraidscompose.tikz}{}{\input{./figures/TikZit/dbraidscompose.tikz}}
\endpgfgraphicnamed &(Bd2)\\
  		  		    & &&\\ 
\beginpgfgraphicnamed{TikZit/dbraidscompose3}
\InputIfFileExists{TikZit/dbraidscompose3.tikz}{}{\input{./figures/TikZit/dbraidscompose3.tikz}}
\endpgfgraphicnamed&(Bd3) &%
\beginpgfgraphicnamed{TikZit/dbraidscompose4}
\InputIfFileExists{TikZit/dbraidscompose4.tikz}{}{\input{./figures/TikZit/dbraidscompose4.tikz}}
\endpgfgraphicnamed &(Bd4)\\
 & &&\\ 
\beginpgfgraphicnamed{TikZit/braidnat1}
\InputIfFileExists{TikZit/braidnat1.tikz}{}{\input{./figures/TikZit/braidnat1.tikz}}
\endpgfgraphicnamed&(Nat1) & %
\beginpgfgraphicnamed{TikZit/braidnat12}
\InputIfFileExists{TikZit/braidnat12.tikz}{}{\input{./figures/TikZit/braidnat12.tikz}}
\endpgfgraphicnamed&(Nat2)\\ 
  & &&\\ 
\beginpgfgraphicnamed{TikZit/braidnat2}
\InputIfFileExists{TikZit/braidnat2.tikz}{}{\input{./figures/TikZit/braidnat2.tikz}}
\endpgfgraphicnamed&(Nat3) &%
\beginpgfgraphicnamed{TikZit/braidnat3}
\InputIfFileExists{TikZit/braidnat3.tikz}{}{\input{./figures/TikZit/braidnat3.tikz}}
\endpgfgraphicnamed &(Nat4)\\
 & &&\\ 
\beginpgfgraphicnamed{TikZit/capbraidnat2}
\InputIfFileExists{TikZit/capbraidnat2.tikz}{}{\input{./figures/TikZit/capbraidnat2.tikz}}
\endpgfgraphicnamed&(Nat5) &%
\beginpgfgraphicnamed{TikZit/cupbraidnat}
\InputIfFileExists{TikZit/cupbraidnat.tikz}{}{\input{./figures/TikZit/cupbraidnat.tikz}}
\endpgfgraphicnamed &(Nat6)\\
 & &&\\ 
  		  		\hline  

  		\end{array}\]      
  	\end{center}
  		\caption{Qudit ZW-calculus rules I, where  $n$ is a positive integer.
}\label{quditrules1}
  \end{figure}
  \FloatBarrier
  
    \begin{figure}[!h]
\begin{center} 
\[
\quad \qquad\begin{array}{|cccc|}
\hline
\beginpgfgraphicnamed{TikZit/wsymmetry}
\InputIfFileExists{TikZit/wsymmetry.tikz}{}{\input{./figures/TikZit/wsymmetry.tikz}}
\endpgfgraphicnamed&(Syb) &%
\beginpgfgraphicnamed{TikZit/wunit}
\InputIfFileExists{TikZit/wunit.tikz}{}{\input{./figures/TikZit/wunit.tikz}}
\endpgfgraphicnamed &(Unt)\\
  & &&\\ 
\beginpgfgraphicnamed{TikZit/wassociativity}
\InputIfFileExists{TikZit/wassociativity.tikz}{}{\input{./figures/TikZit/wassociativity.tikz}}
\endpgfgraphicnamed&(Aso) &%
\beginpgfgraphicnamed{TikZit/blackcopy}
\InputIfFileExists{TikZit/blackcopy.tikz}{}{\input{./figures/TikZit/blackcopy.tikz}}
\endpgfgraphicnamed &(Cpy)\\
    & &&\\ 
\beginpgfgraphicnamed{TikZit/wnaturality}
\InputIfFileExists{TikZit/wnaturality.tikz}{}{\input{./figures/TikZit/wnaturality.tikz}}
\endpgfgraphicnamed&(Wnt1) &%
\beginpgfgraphicnamed{TikZit/wnaturality2}
\InputIfFileExists{TikZit/wnaturality2.tikz}{}{\input{./figures/TikZit/wnaturality2.tikz}}
\endpgfgraphicnamed &(Wnt2)\\
  & &&\\ 
\beginpgfgraphicnamed{TikZit/wbraidsymmetry}
\InputIfFileExists{TikZit/wbraidsymmetry.tikz}{}{\input{./figures/TikZit/wbraidsymmetry.tikz}}
\endpgfgraphicnamed&(Wsm) &  %
\beginpgfgraphicnamed{TikZit/twobraidsymmetry}
\InputIfFileExists{TikZit/twobraidsymmetry.tikz}{}{\input{./figures/TikZit/twobraidsymmetry.tikz}}
\endpgfgraphicnamed&(Bsm) \\
   & &&\\ 
\beginpgfgraphicnamed{TikZit/wdotsempty}
\InputIfFileExists{TikZit/wdotsempty.tikz}{}{\input{./figures/TikZit/wdotsempty.tikz}}
\endpgfgraphicnamed&(Ept) &%
\beginpgfgraphicnamed{TikZit/blackhopfdit}
\InputIfFileExists{TikZit/blackhopfdit.tikz}{}{\input{./figures/TikZit/blackhopfdit.tikz}}
\endpgfgraphicnamed &(Bhf)\\
  		  		\hline  

  		\end{array}\]      
  	\end{center}
  		\caption{Qudit ZW-calculus rules II, where   $\protect\overrightarrow{-1}=\protect\overbrace{(-1,\cdots,-1)}^{d-1}$.
}\label{quditrules2}
  \end{figure}
  \FloatBarrier
  
     \begin{figure}[!h]
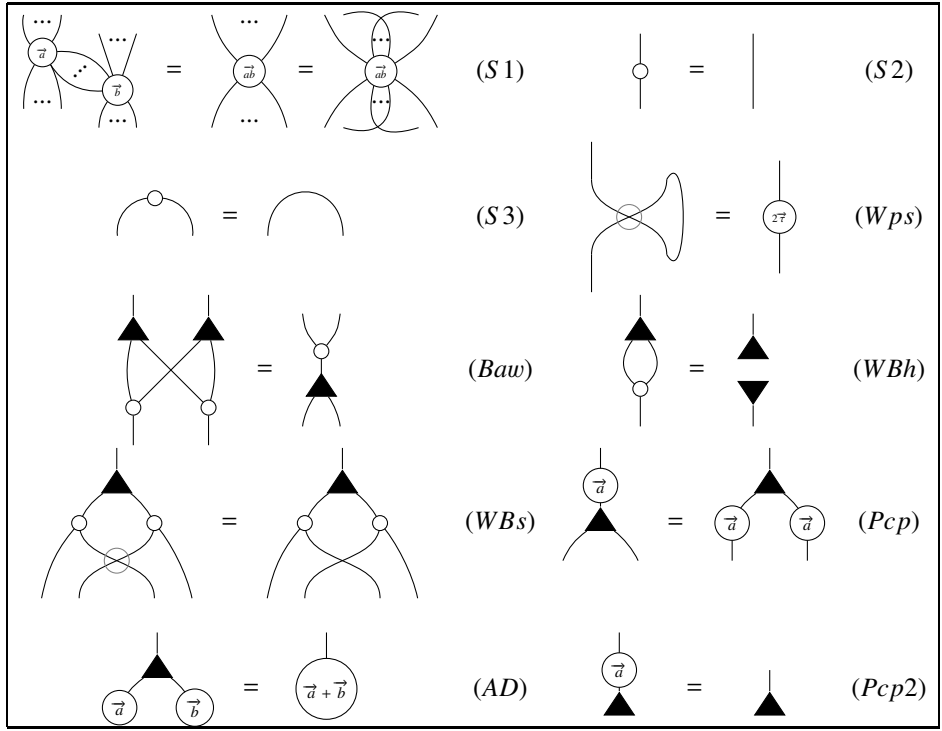

\begin{center} 
\[
\quad \qquad\begin{array}{|cccc|}
\hline
\beginpgfgraphicnamed{TikZit/whitespiderfusiondit}
\InputIfFileExists{TikZit/whitespiderfusiondit.tikz}{}{\input{./figures/TikZit/whitespiderfusiondit.tikz}}
\endpgfgraphicnamed&(S1) &%
\beginpgfgraphicnamed{TikZit/whitespiders2}
\InputIfFileExists{TikZit/whitespiders2.tikz}{}{\input{./figures/TikZit/whitespiders2.tikz}}
\endpgfgraphicnamed &(S2)\\
\beginpgfgraphicnamed{TikZit/whitespiders3}
\InputIfFileExists{TikZit/whitespiders3.tikz}{}{\input{./figures/TikZit/whitespiders3.tikz}}
\endpgfgraphicnamed&(S3) &%
\beginpgfgraphicnamed{TikZit/braidloopwhitephase}
\InputIfFileExists{TikZit/braidloopwhitephase.tikz}{}{\input{./figures/TikZit/braidloopwhitephase.tikz}}
\endpgfgraphicnamed &(Wps)\\
\beginpgfgraphicnamed{TikZit/whiteblackbiagebra}
\InputIfFileExists{TikZit/whiteblackbiagebra.tikz}{}{\input{./figures/TikZit/whiteblackbiagebra.tikz}}
\endpgfgraphicnamed&(Baw) &%
\beginpgfgraphicnamed{TikZit/whiteblackhopf}
\InputIfFileExists{TikZit/whiteblackhopf.tikz}{}{\input{./figures/TikZit/whiteblackhopf.tikz}}
\endpgfgraphicnamed &(WBh)\\
%
\beginpgfgraphicnamed{TikZit/whiteblackswap}
\InputIfFileExists{TikZit/whiteblackswap.tikz}{}{\input{./figures/TikZit/whiteblackswap.tikz}}
\endpgfgraphicnamed&(WBs) &%
\beginpgfgraphicnamed{TikZit/blacknodecopy}
\InputIfFileExists{TikZit/blacknodecopy.tikz}{}{\input{./figures/TikZit/blacknodecopy.tikz}}
\endpgfgraphicnamed &(Pcp)\\
 & &&\\ 
\beginpgfgraphicnamed{TikZit/blackadddit}
\InputIfFileExists{TikZit/blackadddit.tikz}{}{\input{./figures/TikZit/blackadddit.tikz}}
\endpgfgraphicnamed&(AD) &%
\beginpgfgraphicnamed{TikZit/blackcdotphasecopy}
\InputIfFileExists{TikZit/blackcdotphasecopy.tikz}{}{\input{./figures/TikZit/blackcdotphasecopy.tikz}}
\endpgfgraphicnamed &(Pcp2)\\
 \hline  

  		\end{array}\]      
  	\end{center}
  		\caption{Qudit ZW-calculus rules III, where $\protect\overrightarrow{a}=(a_1,\cdots, a_{d-1}), \protect\overrightarrow{b}=(b_1,\cdots, b_{d-1}), \protect\overrightarrow{ab}=(a_1b_1,\cdots, a_{d-1}b_{d-1}), \protect\overrightarrow{a}+\protect\overrightarrow{b}=(a_1+b_1,\cdots, a_{d-1}+b_{d-1}), a_k, b_k\in \mathbb C, k \in \{1,\cdots, d-1\},   \protect\overrightarrow{1}=\protect\overbrace{(1,\cdots,1)}^{d-1}, 2\protect\overrightarrow{\tau}=(\xi,\cdots, \xi^{k^2}, \cdots,  \xi^{(d-1)^2}),  \xi=e^{i\frac{2\pi}{d}}$.
}\label{quditrules3}
  \end{figure}
  \FloatBarrier
  
\begin{remark}  
 These rules are obtained by comparing with the qubit ZW-calculus rules as presented in \cite{amar} and \cite{amarngwanglics}.  
 \end{remark}

\begin{remark}
Note that in general ($d> 2$) there is no bialgebra rule between two the black nodes, i.e.,
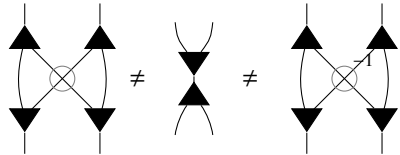
\[ %
\beginpgfgraphicnamed{TikZit/nowbialebra2}
\InputIfFileExists{TikZit/nowbialebra2.tikz}{}{\input{./figures/TikZit/nowbialebra2.tikz}}
\endpgfgraphicnamed\]
 \end{remark} 
 
Based on the associative rule (Aso),  we can define the black spiders and give their interpretation as follows:
\[ %
\beginpgfgraphicnamed{TikZit/mlegsblackspider}
\InputIfFileExists{TikZit/mlegsblackspider.tikz}{}{\input{./figures/TikZit/mlegsblackspider.tikz}}
\endpgfgraphicnamed\quad    %
\beginpgfgraphicnamed{TikZit/mlegsdownblackspider}
\InputIfFileExists{TikZit/mlegsdownblackspider.tikz}{}{\input{./figures/TikZit/mlegsdownblackspider.tikz}}
\endpgfgraphicnamed \]
\[ \left\llbracket %
\beginpgfgraphicnamed{TikZit/mlegsblackspiderone}
\InputIfFileExists{TikZit/mlegsblackspiderone.tikz}{}{\input{./figures/TikZit/mlegsblackspiderone.tikz}}
\endpgfgraphicnamed \right\rrbracket=\underbrace{\ket{0\cdots0}}_{m}\bra{0}+\sum_{i=1}^{d-1}\sum_{k=1}^{m}\overbrace{\ket{\underbrace{0\cdots 0}_{k-1} i 0\cdots 0}}^{m}\bra{i}, 
\left\llbracket %
\beginpgfgraphicnamed{TikZit/mlegsdownblackspiderone}
\InputIfFileExists{TikZit/mlegsdownblackspiderone.tikz}{}{\input{./figures/TikZit/mlegsdownblackspiderone.tikz}}
\endpgfgraphicnamed \right\rrbracket=\ket{0}\underbrace{\bra{0\cdots0}}_{m}+\sum_{i=1}^{d-1}\sum_{k=1}^{m}\ket{i}\overbrace{\bra{\underbrace{0\cdots 0}_{k-1} i 0\cdots 0}}^{m}.
\]

Below we show some properties of the rules by deriving  a few equalities from them.

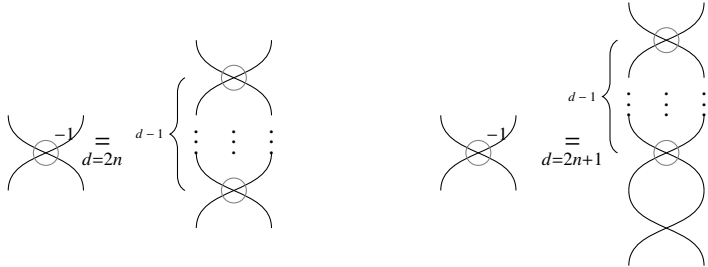
\begin{lemma}\label{crossinvbycorsslm}
\[
\beginpgfgraphicnamed{TikZit/crossinvbycorss}
\InputIfFileExists{TikZit/crossinvbycorss.tikz}{}{\input{./figures/TikZit/crossinvbycorss.tikz}}
\endpgfgraphicnamed 
\]
\end{lemma}
This can be directly obtained from the rules (Bd2) and (Bd3).

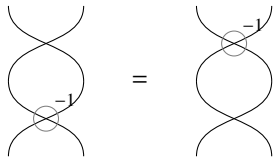
\begin{lemma}\label{dbraidscompose4invlm}
\[
\beginpgfgraphicnamed{TikZit/dbraidscompose4inv}
\InputIfFileExists{TikZit/dbraidscompose4inv.tikz}{}{\input{./figures/TikZit/dbraidscompose4inv.tikz}}
\endpgfgraphicnamed 
\]
\end{lemma}
This can be directly obtained from the rule (Bd4) and lemma \ref{crossinvbycorsslm}.

\begin{lemma}\label{braidptracelm}
\[
\beginpgfgraphicnamed{TikZit/braidptrace}
\InputIfFileExists{TikZit/braidptrace.tikz}{}{\input{./figures/TikZit/braidptrace.tikz}}
\endpgfgraphicnamed
\]
\end{lemma}
\begin{proof}
\[
\beginpgfgraphicnamed{TikZit/braidptraceprf}
\InputIfFileExists{TikZit/braidptraceprf.tikz}{}{\input{./figures/TikZit/braidptraceprf.tikz}}
\endpgfgraphicnamed
\]
\end{proof}

Similarly, we have  
\begin{lemma}
\[
\beginpgfgraphicnamed{TikZit/braidptrace2}
\InputIfFileExists{TikZit/braidptrace2.tikz}{}{\input{./figures/TikZit/braidptrace2.tikz}}
\endpgfgraphicnamed
\]
\end{lemma}

This can be  proved in the same way as that for lemma \ref{braidptracelm}, via lemma  \ref{dbraidscompose4invlm}.

\begin{lemma}
\[
\beginpgfgraphicnamed{TikZit/wivbraidsymmetry}
\InputIfFileExists{TikZit/wivbraidsymmetry.tikz}{}{\input{./figures/TikZit/wivbraidsymmetry.tikz}}
\endpgfgraphicnamed 
\]
\end{lemma}
\begin{proof}
If $d=2n$, then 
\[
\beginpgfgraphicnamed{TikZit/wivbraidsymmetryprf}
\InputIfFileExists{TikZit/wivbraidsymmetryprf.tikz}{}{\input{./figures/TikZit/wivbraidsymmetryprf.tikz}}
\endpgfgraphicnamed 
\]
If $d=2n+1$, then 
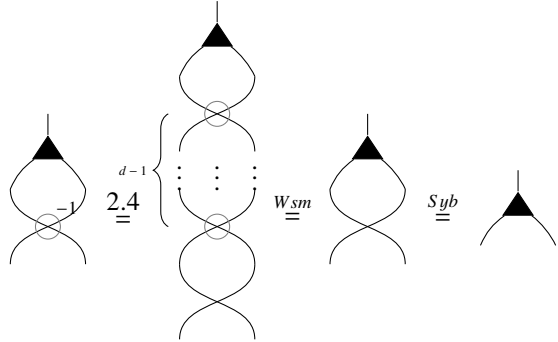
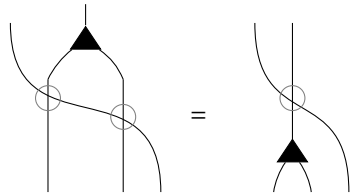
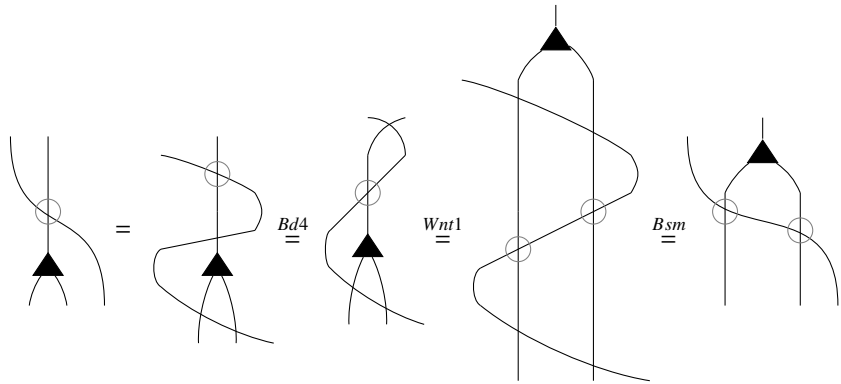
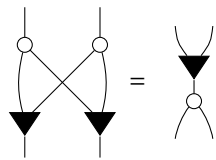
\[
\beginpgfgraphicnamed{TikZit/wivbraidsymmetryprf2}
\InputIfFileExists{TikZit/wivbraidsymmetryprf2.tikz}{}{\input{./figures/TikZit/wivbraidsymmetryprf2.tikz}}
\endpgfgraphicnamed 
\]
\end{proof}

\begin{lemma}
\[
\beginpgfgraphicnamed{TikZit/wnaturalityflip}
\InputIfFileExists{TikZit/wnaturalityflip.tikz}{}{\input{./figures/TikZit/wnaturalityflip.tikz}}
\endpgfgraphicnamed 
\]
\end{lemma}
\begin{proof}
\[
\beginpgfgraphicnamed{TikZit/wnaturalityflipprf}
\InputIfFileExists{TikZit/wnaturalityflipprf.tikz}{}{\input{./figures/TikZit/wnaturalityflipprf.tikz}}
\endpgfgraphicnamed 
\]
\end{proof}

\begin{lemma}\label{whitedownblackbiagebralm}
\[
\beginpgfgraphicnamed{TikZit/whitedownblackbiagebra}
\InputIfFileExists{TikZit/whitedownblackbiagebra.tikz}{}{\input{./figures/TikZit/whitedownblackbiagebra.tikz}}
\endpgfgraphicnamed 
\]
\end{lemma}
\begin{proof}
\[
\beginpgfgraphicnamed{TikZit/whitedownblackbiagebraprf}
\InputIfFileExists{TikZit/whitedownblackbiagebraprf.tikz}{}{\input{./figures/TikZit/whitedownblackbiagebraprf.tikz}}
\endpgfgraphicnamed 
\]
\end{proof}

\begin{lemma}\label{whitedownblackhopflm}
\[
\beginpgfgraphicnamed{TikZit/whitedownblackhopf}
\InputIfFileExists{TikZit/whitedownblackhopf.tikz}{}{\input{./figures/TikZit/whitedownblackhopf.tikz}}
\endpgfgraphicnamed 
\]
\end{lemma}
\begin{proof}
\[
\beginpgfgraphicnamed{TikZit/whitedownblackhopfprf}
\InputIfFileExists{TikZit/whitedownblackhopfprf.tikz}{}{\input{./figures/TikZit/whitedownblackhopfprf.tikz}}
\endpgfgraphicnamed 
\]
\end{proof}

\begin{lemma}
\[
\beginpgfgraphicnamed{TikZit/whiteblackinvswap}
\InputIfFileExists{TikZit/whiteblackinvswap.tikz}{}{\input{./figures/TikZit/whiteblackinvswap.tikz}}
\endpgfgraphicnamed 
\]
\end{lemma}

This follows directly from the rules (Bd4), (WBs) and lemma \ref{crossinvbycorsslm}.

\begin{lemma}
\[
\beginpgfgraphicnamed{TikZit/braidloopphase}
\InputIfFileExists{TikZit/braidloopphase.tikz}{}{\input{./figures/TikZit/braidloopphase.tikz}}
\endpgfgraphicnamed 
\]
\end{lemma}
This follows directly from the rules (S1), (Wps).

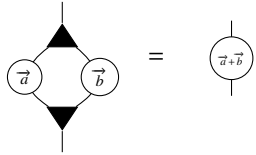
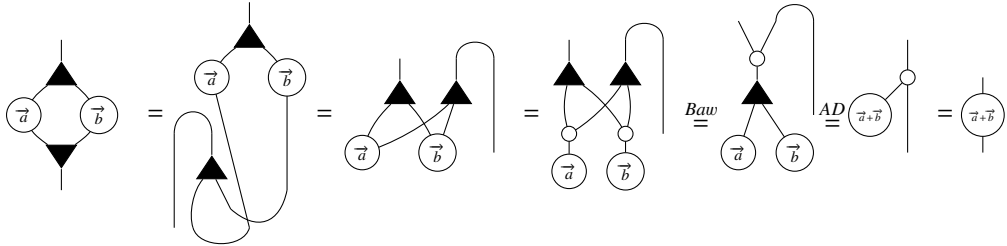
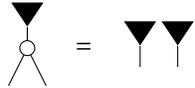
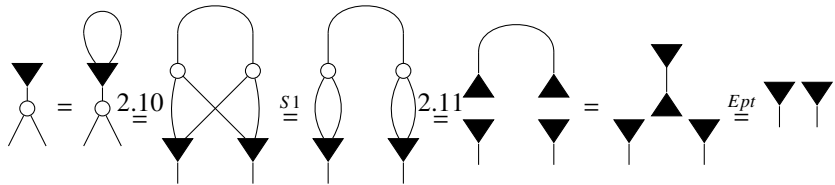
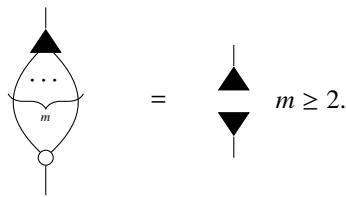
\begin{lemma}
\[
\beginpgfgraphicnamed{TikZit/addition2blacks}
\InputIfFileExists{TikZit/addition2blacks.tikz}{}{\input{./figures/TikZit/addition2blacks.tikz}}
\endpgfgraphicnamed 
\]
\end{lemma}
\begin{proof}
\[
\beginpgfgraphicnamed{TikZit/addition2blacksprf}
\InputIfFileExists{TikZit/addition2blacksprf.tikz}{}{\input{./figures/TikZit/addition2blacksprf.tikz}}
\endpgfgraphicnamed 
\]
\end{proof}

\begin{lemma}
\[
\beginpgfgraphicnamed{TikZit/blackdotcopy}
\InputIfFileExists{TikZit/blackdotcopy.tikz}{}{\input{./figures/TikZit/blackdotcopy.tikz}}
\endpgfgraphicnamed 
\]
\end{lemma}
\begin{proof}
\[
\beginpgfgraphicnamed{TikZit/blackdotcopyprf}
\InputIfFileExists{TikZit/blackdotcopyprf.tikz}{}{\input{./figures/TikZit/blackdotcopyprf.tikz}}
\endpgfgraphicnamed 
\]
\end{proof}

\begin{lemma}
\[
\beginpgfgraphicnamed{TikZit/whiteblackgeneralhopf}
\InputIfFileExists{TikZit/whiteblackgeneralhopf.tikz}{}{\input{./figures/TikZit/whiteblackgeneralhopf.tikz}}
\endpgfgraphicnamed   \quad m \geq 2.
\]
\end{lemma}

It follows directly from the rule (WBh).
\begin{remark}

 
 
Note that in the rule (WBh), the white node and the black node just need two wires to get themselves in disconnected,  while in qudit ZX-calculus, the green (Z) spider and the red (X) spider need $d$ wires to get disconnected. 
 \end{remark}



 

  
\section{Translations between qudit ZX-calculus and non-anyonic qudit ZW-calculus}
Algebraic ZX-calculus \cite{wangalg2020,qwangnormalformbit}  which includes the normal  qubit ZX-calculus  \cite{CoeckeDuncan} has been generalised to any higher dimension  \cite{wangqufinite}.

  First we give the generators of qudit ZX-calculus. 
  \FloatBarrier
 \begin{table}
\begin{center} 
\begin{tabular}{|r@{~}r@{~}c@{~}c|r@{~}r@{~}c@{~}c|}
\hline
$R_{Z,\overrightarrow{a}}^{(n,m)}$&$:$&$n\to m$ & %
\beginpgfgraphicnamed{TikZit/generalgreenspiderqdit2}
\InputIfFileExists{TikZit/generalgreenspiderqdit2.tikz}{}{\input{./figures/TikZit/generalgreenspiderqdit2.tikz}}
\endpgfgraphicnamed  & &&& \\ \hline
$H$&$:$&$1\to 1$ &%
\beginpgfgraphicnamed{TikZit/HadaDecomSingleslt}
\InputIfFileExists{TikZit/HadaDecomSingleslt.tikz}{}{\input{./figures/TikZit/HadaDecomSingleslt.tikz}}
\endpgfgraphicnamed
 &  $H^{\dagger}$&$:$&$ 1\to 1$& %
\beginpgfgraphicnamed{TikZit/RGg_Hadad}
\InputIfFileExists{TikZit/RGg_Hadad.tikz}{}{\input{./figures/TikZit/RGg_Hadad.tikz}}
\endpgfgraphicnamed\\\hline
   $\mathbb I$&$:$&$1\to 1$&%
\beginpgfgraphicnamed{TikZit/Id}
\InputIfFileExists{TikZit/Id.tikz}{}{\input{./figures/TikZit/Id.tikz}}
\endpgfgraphicnamed &   $\sigma$&$:$&$ 2\to 2$& %
\beginpgfgraphicnamed{TikZit/swap}
\InputIfFileExists{TikZit/swap.tikz}{}{\input{./figures/TikZit/swap.tikz}}
\endpgfgraphicnamed\\ \hline
   $C_a$&$:$&$ 0\to 2$& %
\beginpgfgraphicnamed{TikZit/cap}
\InputIfFileExists{TikZit/cap.tikz}{}{\input{./figures/TikZit/cap.tikz}}
\endpgfgraphicnamed &$ C_u$&$:$&$ 2\to 0$&%
\beginpgfgraphicnamed{TikZit/cup}
\InputIfFileExists{TikZit/cup.tikz}{}{\input{./figures/TikZit/cup.tikz}}
\endpgfgraphicnamed \\\hline
 $T_g$&$:$&$1\to 1$&%
\beginpgfgraphicnamed{TikZit/triangle}
\InputIfFileExists{TikZit/triangle.tikz}{}{\input{./figures/TikZit/triangle.tikz}}
\endpgfgraphicnamed  & $T_g^{-1}$&$:$&$1\to 1$&%
\beginpgfgraphicnamed{TikZit/triangleinv}
\InputIfFileExists{TikZit/triangleinv.tikz}{}{\input{./figures/TikZit/triangleinv.tikz}}
\endpgfgraphicnamed \\\hline
\end{tabular}\caption{ Generators of qudit ZX-calculus, where $m,n\in \mathbb N$, $\protect\overrightarrow{a}=(a_1, \cdots, a_{d-1}), a_i \in \mathbb C. $} \label{qbzxgeneratordit}
\end{center}
\end{table}
\FloatBarrier

 For convenience, we make the following denotations: 
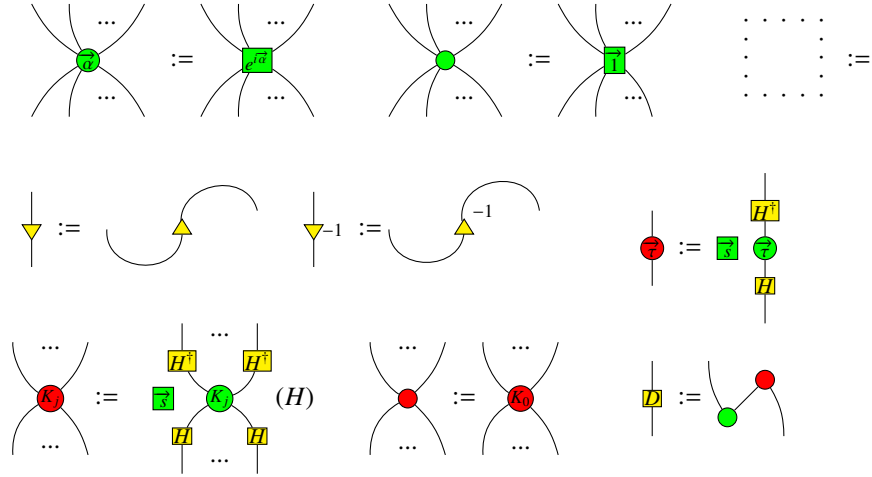
\[
 %
\beginpgfgraphicnamed{TikZit/spidersdenotedit3}
\InputIfFileExists{TikZit/spidersdenotedit3.tikz}{}{\input{./figures/TikZit/spidersdenotedit3.tikz}}
\endpgfgraphicnamed 
\]
 where $   \overrightarrow{\alpha}=(\alpha_1, \cdots, \alpha_{d-1}),  \alpha_i \in \mathbb R,  e^{i\overrightarrow{\alpha}}=(e^{i\alpha_1}, \cdots, e^{i\alpha_{d-1}}), \overrightarrow{1}=\overbrace{(1,\cdots,1)}^{d-1}, \overrightarrow{s}=(0,\cdots,0, \frac{1}{d}-1), \overrightarrow{\tau}=(\tau_1, \cdots, \tau_k, \cdots, \tau_{d-1}), \tau_k=k\pi+\frac{k^2\pi}{d}, 1\leq k \leq d-1,  K_j=(j\frac{2\pi}{d}, 2j\frac{2\pi}{d}, \cdots, (d-1)j\frac{2\pi}{d}),  0\leq j \leq d-1$.
 
 The diagrams of the qudit ZX-calculus have the following standard interpretation $ \left\llbracket \cdot \right\rrbracket$:

 \[
 \left\llbracket %
\beginpgfgraphicnamed{TikZit/generalgreenspiderqdit2}
\InputIfFileExists{TikZit/generalgreenspiderqdit2.tikz}{}{\input{./figures/TikZit/generalgreenspiderqdit2.tikz}}
\endpgfgraphicnamed \right\rrbracket=\sum_{j=0}^{d-1}a_j\ket{j}^{\otimes m}\bra{j}^{\otimes n}, a_0=1,  \protect\overrightarrow{a}=(a_1, \cdots, a_{d-1}), a_i \in \mathbb C. 
\]

 \[
\left\llbracket %
\beginpgfgraphicnamed{TikZit/quditrspiderclassic}
\InputIfFileExists{TikZit/quditrspiderclassic.tikz}{}{\input{./figures/TikZit/quditrspiderclassic.tikz}}
\endpgfgraphicnamed \right\rrbracket=\sum_{\substack{0\leq i_1, \cdots, i_m,  j_1, \cdots, j_n\leq d-1\\ i_1+\cdots+ i_m+j\equiv  j_1+\cdots +j_n(mod~ d)}}\ket{i_1, \cdots, i_m}\bra{j_1, \cdots, j_n}, \]
\[ K_j=(j\frac{2\pi}{d}, 2j\frac{2\pi}{d}, \cdots, (d-1)j\frac{2\pi}{d}),  0\leq j \leq d-1,
\]

\[
\left\llbracket%
\beginpgfgraphicnamed{TikZit/HadaDecomSingleslt}
\begin{tikzpicture}
	\begin{pgfonlayer}{nodelayer}
		\node [style=H box] (0) at (-0.75, 0) {$H$};
		\node [style=none] (1) at (-0.75, -0.5) {};
		\node [style=none] (2) at (-0.75, 0.5) {};
	\end{pgfonlayer}
	\begin{pgfonlayer}{edgelayer}
		\draw (2.center) to (0);
		\draw (1.center) to (0);
	\end{pgfonlayer}
\end{tikzpicture}}
\endpgfgraphicnamed\right\rrbracket=\sum_{k, j=0}^{d-1}\xi^{jk}\ket{j}\bra{k},  \xi=e^{i\frac{2\pi}{d}}, \quad
\left\llbracket%
\beginpgfgraphicnamed{TikZit/RGg_Hadad}
\begin{tikzpicture}
	\begin{pgfonlayer}{nodelayer}
		\node [style={H box}] (0) at (0, -0) {$H^\dagger$};
		\node [style=none] (1) at (0, -0.5) {};
		\node [style=none] (2) at (0, 0.5) {};
	\end{pgfonlayer}
	\begin{pgfonlayer}{edgelayer}
		\draw (2.center) to (0);
		\draw (1.center) to (0);
	\end{pgfonlayer}
\end{tikzpicture}}
\endpgfgraphicnamed\right\rrbracket=\sum_{k, j=0}^{d-1}\bar{\xi}^{jk}\ket{j}\bra{k}, \bar{\xi}=e^{-i\frac{2\pi}{d}},  \]
\[
  \left\llbracket%
\beginpgfgraphicnamed{TikZit/triangle}
\begin{tikzpicture}
	\begin{pgfonlayer}{nodelayer}
		\node [style=none] (0) at (0, 0.5) {};
		\node [style=triangle] (1) at (0, 0) {};
		\node [style=none] (2) at (0, -0.5) {};
	\end{pgfonlayer}
	\begin{pgfonlayer}{edgelayer}
		\draw (0.center) to (2.center);
	\end{pgfonlayer}
\end{tikzpicture}}
\endpgfgraphicnamed\right\rrbracket=I_d+\sum_{i=1}^{d-1}\ket{0}\bra{i}, \quad
   \left\llbracket%
\beginpgfgraphicnamed{TikZit/triangleinv}
\begin{tikzpicture}
	\begin{pgfonlayer}{nodelayer}
		\node [style=none] (0) at (0.25, 0.25) {-{\scriptsize1}};
		\node [style=triangle] (1) at (0, 0) {};
		\node [style=none] (2) at (0, -0.5) {};
		\node [style=none] (3) at (0, 0.5) {};
	\end{pgfonlayer}
	\begin{pgfonlayer}{edgelayer}
		\draw (3.center) to (2.center);
	\end{pgfonlayer}
\end{tikzpicture}}
\endpgfgraphicnamed\right\rrbracket=I_d-\sum_{i=1}^{d-1}\ket{0}\bra{i},
   \]
   \[
      \left\llbracket%
\beginpgfgraphicnamed{TikZit/redtaugate}
\begin{tikzpicture}
	\begin{pgfonlayer}{nodelayer}
		\node [style=none] (0) at (0, -0.5) {};
		\node [style=none] (1) at (0, 0.5) {};
		\node [style=rn] (2) at (0, 0) { ${\scriptstyle  \overrightarrow{\tau}}$ };
	\end{pgfonlayer}
	\begin{pgfonlayer}{edgelayer}
		\draw (2) to (1.center);
		\draw (2) to (0.center);
	\end{pgfonlayer}
\end{tikzpicture}}
\endpgfgraphicnamed\right\rrbracket=\frac{1}{d}\sum_{k,l,n=0}^{d-1}e^{i\tau_l}\xi^{(k-n)l}\ket{k}\bra{n}, \overrightarrow{\tau}=(\tau_1, \cdots, \tau_k, \cdots, \tau_{d-1}), \tau_k=k\pi+\frac{k^2\pi}{d}, 0\leq k \leq d-1,
   \]

\[
   \left\llbracket%
\beginpgfgraphicnamed{TikZit/Id}
\begin{tikzpicture}
	\begin{pgfonlayer}{nodelayer}
		\node [style=none] (1) at (0.5, 0.3) {};
		\node [style=none] (2) at (0.5, -0.3) {};
		\node [style=none] (3) at (0.5, -0.5) {};
		\node [style=none] (4) at (0.5, 0.5) {};
	\end{pgfonlayer}
	\begin{pgfonlayer}{edgelayer}
		\draw (1.center) to (2.center);
	\end{pgfonlayer}
\end{tikzpicture}}
\endpgfgraphicnamed\right\rrbracket=I_d=\sum_{j=0}^{d-1}\ket{j}\bra{j},  \quad
 \left\llbracket%
\beginpgfgraphicnamed{TikZit/swap}
\InputIfFileExists{TikZit/swap.tikz}{}{\input{./figures/TikZit/swap.tikz}}
\endpgfgraphicnamed\right\rrbracket=\sum_{i, j=0}^{d-1}\ket{ji}\bra{ij},\quad
  \left\llbracket%
\beginpgfgraphicnamed{TikZit/cap}
\begin{tikzpicture}
	\begin{pgfonlayer}{nodelayer}
		\node [style=none] (0) at (-0.5, -0.25) {};
		\node [style=none] (1) at (0.5, -0.25) {};
	\end{pgfonlayer}
	\begin{pgfonlayer}{edgelayer}
		\draw [bend left=90, looseness=1.50] (0.center) to (1.center);
	\end{pgfonlayer}
\end{tikzpicture}}
\endpgfgraphicnamed\right\rrbracket=\sum_{j=0}^{d-1}\ket{jj}, \quad
   \left\llbracket%
\beginpgfgraphicnamed{TikZit/cup}
\begin{tikzpicture}
	\begin{pgfonlayer}{nodelayer}
		\node [style=none] (0) at (-0.5, 0.25) {};
		\node [style=none] (1) at (0.5, 0.25) {};
	\end{pgfonlayer}
	\begin{pgfonlayer}{edgelayer}
		\draw [bend right=90, looseness=1.50] (0.center) to (1.center);
	\end{pgfonlayer}
\end{tikzpicture}}
\endpgfgraphicnamed\right\rrbracket=\sum_{j=0}^{d-1}\bra{jj},   \quad \left\llbracket%
\beginpgfgraphicnamed{TikZit/emptysquare}
\InputIfFileExists{TikZit/emptysquare.tikz}{}{\input{./figures/TikZit/emptysquare.tikz}}
\endpgfgraphicnamed\right\rrbracket=1,
      \]

\[  \llbracket D_1\otimes D_2  \rrbracket =  \llbracket D_1  \rrbracket \otimes  \llbracket  D_2  \rrbracket, \quad 
 \llbracket D_1\circ D_2  \rrbracket =  \llbracket D_1  \rrbracket \circ  \llbracket  D_2  \rrbracket.
  \]

As similar to the qubit case \cite{amarngwanglics}, there is a bidirectional  translation between the non-anyonic qudit ZW-calculus and the qudit ZX-calculus \cite{wangqufinite}.  The translation from ZX to ZW is denoted as $ \left\llbracket \cdot \right\rrbracket_{XW}$. Then we have 

\[
  \left\llbracket%
\beginpgfgraphicnamed{TikZit/Id}
\begin{tikzpicture}
	\begin{pgfonlayer}{nodelayer}
		\node [style=none] (1) at (0.5, 0.3) {};
		\node [style=none] (2) at (0.5, -0.3) {};
		\node [style=none] (3) at (0.5, -0.5) {};
		\node [style=none] (4) at (0.5, 0.5) {};
	\end{pgfonlayer}
	\begin{pgfonlayer}{edgelayer}
		\draw (1.center) to (2.center);
	\end{pgfonlayer}
\end{tikzpicture}}
\endpgfgraphicnamed\right\rrbracket_{XW}=%
\beginpgfgraphicnamed{TikZit/Id}
\begin{tikzpicture}
	\begin{pgfonlayer}{nodelayer}
		\node [style=none] (1) at (0.5, 0.3) {};
		\node [style=none] (2) at (0.5, -0.3) {};
		\node [style=none] (3) at (0.5, -0.5) {};
		\node [style=none] (4) at (0.5, 0.5) {};
	\end{pgfonlayer}
	\begin{pgfonlayer}{edgelayer}
		\draw (1.center) to (2.center);
	\end{pgfonlayer}
\end{tikzpicture}}
\endpgfgraphicnamed,  \quad
 \left\llbracket%
\beginpgfgraphicnamed{TikZit/swap}
\InputIfFileExists{TikZit/swap.tikz}{}{\input{./figures/TikZit/swap.tikz}}
\endpgfgraphicnamed\right\rrbracket_{XW}=%
\beginpgfgraphicnamed{TikZit/swap}
\InputIfFileExists{TikZit/swap.tikz}{}{\input{./figures/TikZit/swap.tikz}}
\endpgfgraphicnamed,\quad
  \left\llbracket%
\beginpgfgraphicnamed{TikZit/cap}
\begin{tikzpicture}
	\begin{pgfonlayer}{nodelayer}
		\node [style=none] (0) at (-0.5, -0.25) {};
		\node [style=none] (1) at (0.5, -0.25) {};
	\end{pgfonlayer}
	\begin{pgfonlayer}{edgelayer}
		\draw [bend left=90, looseness=1.50] (0.center) to (1.center);
	\end{pgfonlayer}
\end{tikzpicture}}
\endpgfgraphicnamed\right\rrbracket_{XW}=%
\beginpgfgraphicnamed{TikZit/cap}
\begin{tikzpicture}
	\begin{pgfonlayer}{nodelayer}
		\node [style=none] (0) at (-0.5, -0.25) {};
		\node [style=none] (1) at (0.5, -0.25) {};
	\end{pgfonlayer}
	\begin{pgfonlayer}{edgelayer}
		\draw [bend left=90, looseness=1.50] (0.center) to (1.center);
	\end{pgfonlayer}
\end{tikzpicture}}
\endpgfgraphicnamed, \quad
   \left\llbracket%
\beginpgfgraphicnamed{TikZit/cup}
\begin{tikzpicture}
	\begin{pgfonlayer}{nodelayer}
		\node [style=none] (0) at (-0.5, 0.25) {};
		\node [style=none] (1) at (0.5, 0.25) {};
	\end{pgfonlayer}
	\begin{pgfonlayer}{edgelayer}
		\draw [bend right=90, looseness=1.50] (0.center) to (1.center);
	\end{pgfonlayer}
\end{tikzpicture}}
\endpgfgraphicnamed\right\rrbracket_{XW}=%
\beginpgfgraphicnamed{TikZit/cup}
\begin{tikzpicture}
	\begin{pgfonlayer}{nodelayer}
		\node [style=none] (0) at (-0.5, 0.25) {};
		\node [style=none] (1) at (0.5, 0.25) {};
	\end{pgfonlayer}
	\begin{pgfonlayer}{edgelayer}
		\draw [bend right=90, looseness=1.50] (0.center) to (1.center);
	\end{pgfonlayer}
\end{tikzpicture}}
\endpgfgraphicnamed,  
\]

\[
  \left\llbracket%
\beginpgfgraphicnamed{TikZit/generalgreenspiderqdit2}
\InputIfFileExists{TikZit/generalgreenspiderqdit2.tikz}{}{\input{./figures/TikZit/generalgreenspiderqdit2.tikz}}
\endpgfgraphicnamed\right\rrbracket_{XW}= %
\beginpgfgraphicnamed{TikZit/whitespiderdit}
\InputIfFileExists{TikZit/whitespiderdit.tikz}{}{\input{./figures/TikZit/whitespiderdit.tikz}}
\endpgfgraphicnamed, \quad \left\llbracket%
\beginpgfgraphicnamed{TikZit/HadaDecomSingleslt}
\begin{tikzpicture}
	\begin{pgfonlayer}{nodelayer}
		\node [style=H box] (0) at (-0.75, 0) {$H$};
		\node [style=none] (1) at (-0.75, -0.5) {};
		\node [style=none] (2) at (-0.75, 0.5) {};
	\end{pgfonlayer}
	\begin{pgfonlayer}{edgelayer}
		\draw (2.center) to (0);
		\draw (1.center) to (0);
	\end{pgfonlayer}
\end{tikzpicture}}
\endpgfgraphicnamed\right\rrbracket_{XW}= %
\beginpgfgraphicnamed{TikZit/htranslate}
\InputIfFileExists{TikZit/htranslate.tikz}{}{\input{./figures/TikZit/htranslate.tikz}}
\endpgfgraphicnamed, \quad
\left\llbracket%
\beginpgfgraphicnamed{TikZit/RGg_Hadad}
\begin{tikzpicture}
	\begin{pgfonlayer}{nodelayer}
		\node [style={H box}] (0) at (0, -0) {$H^\dagger$};
		\node [style=none] (1) at (0, -0.5) {};
		\node [style=none] (2) at (0, 0.5) {};
	\end{pgfonlayer}
	\begin{pgfonlayer}{edgelayer}
		\draw (2.center) to (0);
		\draw (1.center) to (0);
	\end{pgfonlayer}
\end{tikzpicture}}
\endpgfgraphicnamed\right\rrbracket_{XW}= %
\beginpgfgraphicnamed{TikZit/hinversetranslate}
\InputIfFileExists{TikZit/hinversetranslate.tikz}{}{\input{./figures/TikZit/hinversetranslate.tikz}}
\endpgfgraphicnamed,  \quad  \left\llbracket%
\beginpgfgraphicnamed{TikZit/emptysquare}
\InputIfFileExists{TikZit/emptysquare.tikz}{}{\input{./figures/TikZit/emptysquare.tikz}}
\endpgfgraphicnamed\right\rrbracket_{XW}=%
\beginpgfgraphicnamed{TikZit/emptysquare}
\InputIfFileExists{TikZit/emptysquare.tikz}{}{\input{./figures/TikZit/emptysquare.tikz}}
\endpgfgraphicnamed,
\]

\[
  \left\llbracket%
\beginpgfgraphicnamed{TikZit/triangle}
\begin{tikzpicture}
	\begin{pgfonlayer}{nodelayer}
		\node [style=none] (0) at (0, 0.5) {};
		\node [style=triangle] (1) at (0, 0) {};
		\node [style=none] (2) at (0, -0.5) {};
	\end{pgfonlayer}
	\begin{pgfonlayer}{edgelayer}
		\draw (0.center) to (2.center);
	\end{pgfonlayer}
\end{tikzpicture}}
\endpgfgraphicnamed\right\rrbracket_{XW}= %
\beginpgfgraphicnamed{TikZit/triangletranslate}
\InputIfFileExists{TikZit/triangletranslate.tikz}{}{\input{./figures/TikZit/triangletranslate.tikz}}
\endpgfgraphicnamed, \quad
\left\llbracket%
\beginpgfgraphicnamed{TikZit/triangleinv}
\begin{tikzpicture}
	\begin{pgfonlayer}{nodelayer}
		\node [style=none] (0) at (0.25, 0.25) {-{\scriptsize1}};
		\node [style=triangle] (1) at (0, 0) {};
		\node [style=none] (2) at (0, -0.5) {};
		\node [style=none] (3) at (0, 0.5) {};
	\end{pgfonlayer}
	\begin{pgfonlayer}{edgelayer}
		\draw (3.center) to (2.center);
	\end{pgfonlayer}
\end{tikzpicture}}
\endpgfgraphicnamed\right\rrbracket_{XW}= %
\beginpgfgraphicnamed{TikZit/invtriangletranslate}
\InputIfFileExists{TikZit/invtriangletranslate.tikz}{}{\input{./figures/TikZit/invtriangletranslate.tikz}}
\endpgfgraphicnamed, \quad  \llbracket D_1\otimes D_2  \rrbracket_{XW} =  \llbracket D_1  \rrbracket \otimes  \llbracket  D_2  \rrbracket, \quad 
 \llbracket D_1\circ D_2  \rrbracket_{XW} =  \llbracket D_1  \rrbracket \circ  \llbracket  D_2  \rrbracket,
\]
 where $\overrightarrow{-1}= \overbrace{(-1,\cdots,-1)}^{d-1}$.
 
 The translation $ \left\llbracket \cdot \right\rrbracket_{XW}$ preserves the standard interpretation:
 \begin{lemma}\label{xtowpreserve}
 For any diagram $D$ produced from the qudit ZX-calculus, we have $ \left\llbracket  \left\llbracket D \right\rrbracket_{XW}\right\rrbracket= \left\llbracket D \right\rrbracket$.
 \end{lemma}
 
  \begin{proof}
  Since each diagram is generated by generators, it suffices to prove the lemma only for all the generators. This is a routine check for which we omit the verification  details here. 
   \end{proof}
   
  Now we give a translation in another direction: from the non-anyonic qudit ZW-calculus to the qudit ZX-calculus which is denoted as $ \left\llbracket \cdot \right\rrbracket_{WX}$. 
   
 \[
  \left\llbracket%
\beginpgfgraphicnamed{TikZit/Id}
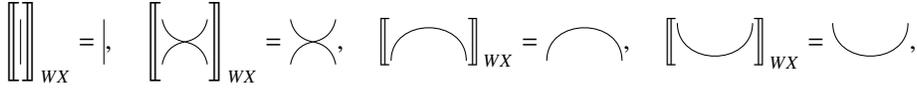
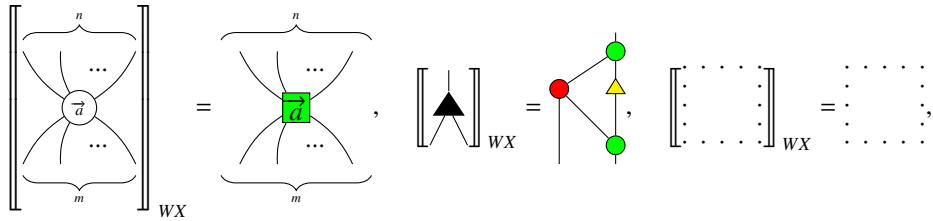
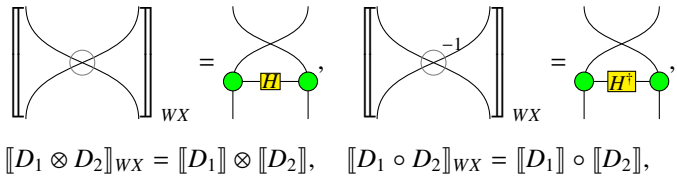
\begin{tikzpicture}
	\begin{pgfonlayer}{nodelayer}
		\node [style=none] (1) at (0.5, 0.3) {};
		\node [style=none] (2) at (0.5, -0.3) {};
		\node [style=none] (3) at (0.5, -0.5) {};
		\node [style=none] (4) at (0.5, 0.5) {};
	\end{pgfonlayer}
	\begin{pgfonlayer}{edgelayer}
		\draw (1.center) to (2.center);
	\end{pgfonlayer}
\end{tikzpicture}}
\endpgfgraphicnamed\right\rrbracket_{WX}=%
\beginpgfgraphicnamed{TikZit/Id}
\begin{tikzpicture}
	\begin{pgfonlayer}{nodelayer}
		\node [style=none] (1) at (0.5, 0.3) {};
		\node [style=none] (2) at (0.5, -0.3) {};
		\node [style=none] (3) at (0.5, -0.5) {};
		\node [style=none] (4) at (0.5, 0.5) {};
	\end{pgfonlayer}
	\begin{pgfonlayer}{edgelayer}
		\draw (1.center) to (2.center);
	\end{pgfonlayer}
\end{tikzpicture}}
\endpgfgraphicnamed,  \quad
 \left\llbracket%
\beginpgfgraphicnamed{TikZit/swap}
\InputIfFileExists{TikZit/swap.tikz}{}{\input{./figures/TikZit/swap.tikz}}
\endpgfgraphicnamed\right\rrbracket_{WX}=%
\beginpgfgraphicnamed{TikZit/swap}
\InputIfFileExists{TikZit/swap.tikz}{}{\input{./figures/TikZit/swap.tikz}}
\endpgfgraphicnamed,\quad
  \left\llbracket%
\beginpgfgraphicnamed{TikZit/cap}
\begin{tikzpicture}
	\begin{pgfonlayer}{nodelayer}
		\node [style=none] (0) at (-0.5, -0.25) {};
		\node [style=none] (1) at (0.5, -0.25) {};
	\end{pgfonlayer}
	\begin{pgfonlayer}{edgelayer}
		\draw [bend left=90, looseness=1.50] (0.center) to (1.center);
	\end{pgfonlayer}
\end{tikzpicture}}
\endpgfgraphicnamed\right\rrbracket_{WX}=%
\beginpgfgraphicnamed{TikZit/cap}
\begin{tikzpicture}
	\begin{pgfonlayer}{nodelayer}
		\node [style=none] (0) at (-0.5, -0.25) {};
		\node [style=none] (1) at (0.5, -0.25) {};
	\end{pgfonlayer}
	\begin{pgfonlayer}{edgelayer}
		\draw [bend left=90, looseness=1.50] (0.center) to (1.center);
	\end{pgfonlayer}
\end{tikzpicture}}
\endpgfgraphicnamed, \quad
   \left\llbracket%
\beginpgfgraphicnamed{TikZit/cup}
\begin{tikzpicture}
	\begin{pgfonlayer}{nodelayer}
		\node [style=none] (0) at (-0.5, 0.25) {};
		\node [style=none] (1) at (0.5, 0.25) {};
	\end{pgfonlayer}
	\begin{pgfonlayer}{edgelayer}
		\draw [bend right=90, looseness=1.50] (0.center) to (1.center);
	\end{pgfonlayer}
\end{tikzpicture}}
\endpgfgraphicnamed\right\rrbracket_{WX}=%
\beginpgfgraphicnamed{TikZit/cup}
\begin{tikzpicture}
	\begin{pgfonlayer}{nodelayer}
		\node [style=none] (0) at (-0.5, 0.25) {};
		\node [style=none] (1) at (0.5, 0.25) {};
	\end{pgfonlayer}
	\begin{pgfonlayer}{edgelayer}
		\draw [bend right=90, looseness=1.50] (0.center) to (1.center);
	\end{pgfonlayer}
\end{tikzpicture}}
\endpgfgraphicnamed,  
\]  

\[
  \left\llbracket%
\beginpgfgraphicnamed{TikZit/whitespiderdit}
\InputIfFileExists{TikZit/whitespiderdit.tikz}{}{\input{./figures/TikZit/whitespiderdit.tikz}}
\endpgfgraphicnamed\right\rrbracket_{WX}= %
\beginpgfgraphicnamed{TikZit/generalgreenspiderqdit2}
\InputIfFileExists{TikZit/generalgreenspiderqdit2.tikz}{}{\input{./figures/TikZit/generalgreenspiderqdit2.tikz}}
\endpgfgraphicnamed, \quad \left\llbracket%
\beginpgfgraphicnamed{TikZit/blacktriangledit}
\begin{tikzpicture}
	\begin{pgfonlayer}{nodelayer}
		\node [style=none] (0) at (0.25, -0.5) {};
		\node [style=none] (1) at (0, 0.5) {};
		\node [style=none] (2) at (-0.25, -0.5) {};
		\node [style=dbspider] (3) at (0, 0) {};
	\end{pgfonlayer}
	\begin{pgfonlayer}{edgelayer}
		\draw (1.center) to (3);
		\draw (3) to (2.center);
		\draw (3) to (0.center);
	\end{pgfonlayer}
\end{tikzpicture}}
\endpgfgraphicnamed\right\rrbracket_{WX}= %
\beginpgfgraphicnamed{TikZit/winzxdit}
\InputIfFileExists{TikZit/winzxdit.tikz}{}{\input{./figures/TikZit/winzxdit.tikz}}
\endpgfgraphicnamed, \quad
  \left\llbracket%
\beginpgfgraphicnamed{TikZit/emptysquare}
\InputIfFileExists{TikZit/emptysquare.tikz}{}{\input{./figures/TikZit/emptysquare.tikz}}
\endpgfgraphicnamed\right\rrbracket_{WX}=%
\beginpgfgraphicnamed{TikZit/emptysquare}
\InputIfFileExists{TikZit/emptysquare.tikz}{}{\input{./figures/TikZit/emptysquare.tikz}}
\endpgfgraphicnamed,
\]

\[
  \left\llbracket%
\beginpgfgraphicnamed{TikZit/braid}
\InputIfFileExists{TikZit/braid.tikz}{}{\input{./figures/TikZit/braid.tikz}}
\endpgfgraphicnamed\right\rrbracket_{WX}= %
\beginpgfgraphicnamed{TikZit/crossinzxdit}
\InputIfFileExists{TikZit/crossinzxdit.tikz}{}{\input{./figures/TikZit/crossinzxdit.tikz}}
\endpgfgraphicnamed, \quad
\left\llbracket%
\beginpgfgraphicnamed{TikZit/braidd2}
\InputIfFileExists{TikZit/braidd2.tikz}{}{\input{./figures/TikZit/braidd2.tikz}}
\endpgfgraphicnamed\right\rrbracket_{WX}= %
\beginpgfgraphicnamed{TikZit/crossinverseinzxdit}
\InputIfFileExists{TikZit/crossinverseinzxdit.tikz}{}{\input{./figures/TikZit/crossinverseinzxdit.tikz}}
\endpgfgraphicnamed, \]
\[
 \llbracket D_1\otimes D_2  \rrbracket_{WX} =  \llbracket D_1  \rrbracket \otimes  \llbracket  D_2  \rrbracket, \quad 
 \llbracket D_1\circ D_2  \rrbracket_{WX} =  \llbracket D_1  \rrbracket \circ  \llbracket  D_2  \rrbracket,
\]
 where $\overrightarrow{a}=(a_1, \cdots, a_{d-1})$.
 
Similar to lemma \ref{xtowpreserve}, the  translation $ \left\llbracket \cdot \right\rrbracket_{WX}$ also preserves the standard interpretation:

 \begin{lemma}\label{wtoxpreserve}
 For any diagram $D$ produced from the qudit ZW-calculus, we have $ \left\llbracket  \left\llbracket D \right\rrbracket_{WX}\right\rrbracket= \left\llbracket D \right\rrbracket$.
 \end{lemma}
 
 From ZX diagrams to ZW diagrams and then back, you got the same original diagrams, i.e., it is an isometry. 
  \begin{lemma}\label{zxtozwtozx}
 For any diagram $D$ produced from the qudit ZX-calculus, we have $ \left\llbracket  \left\llbracket D \right\rrbracket_{XW}\right\rrbracket_{WX}=D$.
 \end{lemma}
   \begin{proof}
 Since both $ \left\llbracket \cdot \right\rrbracket_{XW}$ and $ \left\llbracket \cdot \right\rrbracket_{WX}$ preserve the sequential and tensor compositions,  we only need to prove for all the generators. Also the identity $\mathbb I$, swap $\sigma$, cap  $C_a$ and cup $C_u$ are the common generators of ZX-calculus and ZW-calculus, so we just need to consider the remaining generators of ZX-calculus.
 
  \[ 
  \left\llbracket  \left\llbracket %
\beginpgfgraphicnamed{TikZit/generalgreenspiderqdit2}
\InputIfFileExists{TikZit/generalgreenspiderqdit2.tikz}{}{\input{./figures/TikZit/generalgreenspiderqdit2.tikz}}
\endpgfgraphicnamed  \right\rrbracket_{XW}\right\rrbracket_{WX}=  \left\llbracket %
\beginpgfgraphicnamed{TikZit/whitespiderdit}
\InputIfFileExists{TikZit/whitespiderdit.tikz}{}{\input{./figures/TikZit/whitespiderdit.tikz}}
\endpgfgraphicnamed \right\rrbracket_{WX}= %
\beginpgfgraphicnamed{TikZit/generalgreenspiderqdit2}
\InputIfFileExists{TikZit/generalgreenspiderqdit2.tikz}{}{\input{./figures/TikZit/generalgreenspiderqdit2.tikz}}
\endpgfgraphicnamed  
  \]
 
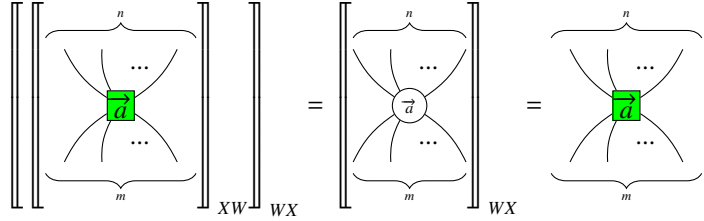
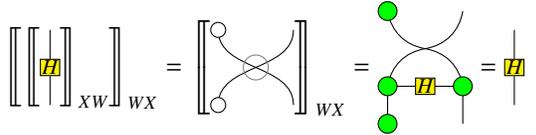
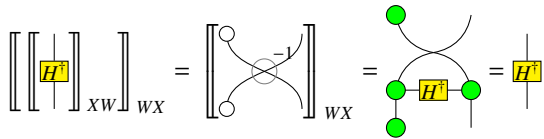
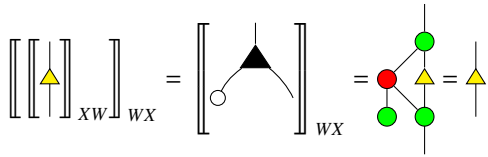
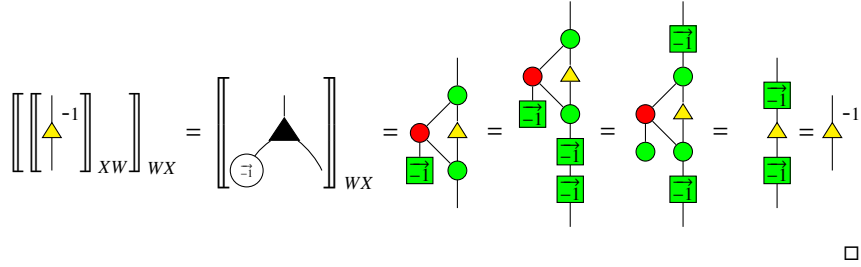
   \[ 
  \left\llbracket  \left\llbracket %
\beginpgfgraphicnamed{TikZit/HadaDecomSingleslt}
\begin{tikzpicture}
	\begin{pgfonlayer}{nodelayer}
		\node [style=H box] (0) at (-0.75, 0) {$H$};
		\node [style=none] (1) at (-0.75, -0.5) {};
		\node [style=none] (2) at (-0.75, 0.5) {};
	\end{pgfonlayer}
	\begin{pgfonlayer}{edgelayer}
		\draw (2.center) to (0);
		\draw (1.center) to (0);
	\end{pgfonlayer}
\end{tikzpicture}}
\endpgfgraphicnamed  \right\rrbracket_{XW}\right\rrbracket_{WX}=  \left\llbracket %
\beginpgfgraphicnamed{TikZit/htranslate}
\InputIfFileExists{TikZit/htranslate.tikz}{}{\input{./figures/TikZit/htranslate.tikz}}
\endpgfgraphicnamed \right\rrbracket_{WX}= %
\beginpgfgraphicnamed{TikZit/Hadamardwhiteint2}
\InputIfFileExists{TikZit/Hadamardwhiteint2.tikz}{}{\input{./figures/TikZit/Hadamardwhiteint2.tikz}}
\endpgfgraphicnamed =  %
\beginpgfgraphicnamed{TikZit/HadaDecomSingleslt}
\begin{tikzpicture}
	\begin{pgfonlayer}{nodelayer}
		\node [style=H box] (0) at (-0.75, 0) {$H$};
		\node [style=none] (1) at (-0.75, -0.5) {};
		\node [style=none] (2) at (-0.75, 0.5) {};
	\end{pgfonlayer}
	\begin{pgfonlayer}{edgelayer}
		\draw (2.center) to (0);
		\draw (1.center) to (0);
	\end{pgfonlayer}
\end{tikzpicture}}
\endpgfgraphicnamed 
  \]

   \[ 
  \left\llbracket  \left\llbracket %
\beginpgfgraphicnamed{TikZit/RGg_Hadad}
\begin{tikzpicture}
	\begin{pgfonlayer}{nodelayer}
		\node [style={H box}] (0) at (0, -0) {$H^\dagger$};
		\node [style=none] (1) at (0, -0.5) {};
		\node [style=none] (2) at (0, 0.5) {};
	\end{pgfonlayer}
	\begin{pgfonlayer}{edgelayer}
		\draw (2.center) to (0);
		\draw (1.center) to (0);
	\end{pgfonlayer}
\end{tikzpicture}}
\endpgfgraphicnamed  \right\rrbracket_{XW}\right\rrbracket_{WX}=  \left\llbracket %
\beginpgfgraphicnamed{TikZit/hinversetranslate}
\InputIfFileExists{TikZit/hinversetranslate.tikz}{}{\input{./figures/TikZit/hinversetranslate.tikz}}
\endpgfgraphicnamed \right\rrbracket_{WX}= %
\beginpgfgraphicnamed{TikZit/Hadamardinvwhiteint}
\InputIfFileExists{TikZit/Hadamardinvwhiteint.tikz}{}{\input{./figures/TikZit/Hadamardinvwhiteint.tikz}}
\endpgfgraphicnamed =  %
\beginpgfgraphicnamed{TikZit/RGg_Hadad}
\begin{tikzpicture}
	\begin{pgfonlayer}{nodelayer}
		\node [style={H box}] (0) at (0, -0) {$H^\dagger$};
		\node [style=none] (1) at (0, -0.5) {};
		\node [style=none] (2) at (0, 0.5) {};
	\end{pgfonlayer}
	\begin{pgfonlayer}{edgelayer}
		\draw (2.center) to (0);
		\draw (1.center) to (0);
	\end{pgfonlayer}
\end{tikzpicture}}
\endpgfgraphicnamed 
  \]
 
    \[ 
  \left\llbracket  \left\llbracket %
\beginpgfgraphicnamed{TikZit/triangle}
\begin{tikzpicture}
	\begin{pgfonlayer}{nodelayer}
		\node [style=none] (0) at (0, 0.5) {};
		\node [style=triangle] (1) at (0, 0) {};
		\node [style=none] (2) at (0, -0.5) {};
	\end{pgfonlayer}
	\begin{pgfonlayer}{edgelayer}
		\draw (0.center) to (2.center);
	\end{pgfonlayer}
\end{tikzpicture}}
\endpgfgraphicnamed  \right\rrbracket_{XW}\right\rrbracket_{WX}=  \left\llbracket %
\beginpgfgraphicnamed{TikZit/triangletranslate}
\InputIfFileExists{TikZit/triangletranslate.tikz}{}{\input{./figures/TikZit/triangletranslate.tikz}}
\endpgfgraphicnamed \right\rrbracket_{WX}= %
\beginpgfgraphicnamed{TikZit/tranglexwx}
\InputIfFileExists{TikZit/tranglexwx.tikz}{}{\input{./figures/TikZit/tranglexwx.tikz}}
\endpgfgraphicnamed =  %
\beginpgfgraphicnamed{TikZit/triangle}
\begin{tikzpicture}
	\begin{pgfonlayer}{nodelayer}
		\node [style=none] (0) at (0, 0.5) {};
		\node [style=triangle] (1) at (0, 0) {};
		\node [style=none] (2) at (0, -0.5) {};
	\end{pgfonlayer}
	\begin{pgfonlayer}{edgelayer}
		\draw (0.center) to (2.center);
	\end{pgfonlayer}
\end{tikzpicture}}
\endpgfgraphicnamed 
  \]
    \[ 
  \left\llbracket  \left\llbracket %
\beginpgfgraphicnamed{TikZit/triangleinv}
\begin{tikzpicture}
	\begin{pgfonlayer}{nodelayer}
		\node [style=none] (0) at (0.25, 0.25) {-{\scriptsize1}};
		\node [style=triangle] (1) at (0, 0) {};
		\node [style=none] (2) at (0, -0.5) {};
		\node [style=none] (3) at (0, 0.5) {};
	\end{pgfonlayer}
	\begin{pgfonlayer}{edgelayer}
		\draw (3.center) to (2.center);
	\end{pgfonlayer}
\end{tikzpicture}}
\endpgfgraphicnamed  \right\rrbracket_{XW}\right\rrbracket_{WX}=  \left\llbracket %
\beginpgfgraphicnamed{TikZit/invtriangletranslate}
\InputIfFileExists{TikZit/invtriangletranslate.tikz}{}{\input{./figures/TikZit/invtriangletranslate.tikz}}
\endpgfgraphicnamed \right\rrbracket_{WX}= %
\beginpgfgraphicnamed{TikZit/trangleinvxwx}
\InputIfFileExists{TikZit/trangleinvxwx.tikz}{}{\input{./figures/TikZit/trangleinvxwx.tikz}}
\endpgfgraphicnamed =  %
\beginpgfgraphicnamed{TikZit/triangleinv}
\begin{tikzpicture}
	\begin{pgfonlayer}{nodelayer}
		\node [style=none] (0) at (0.25, 0.25) {-{\scriptsize1}};
		\node [style=triangle] (1) at (0, 0) {};
		\node [style=none] (2) at (0, -0.5) {};
		\node [style=none] (3) at (0, 0.5) {};
	\end{pgfonlayer}
	\begin{pgfonlayer}{edgelayer}
		\draw (3.center) to (2.center);
	\end{pgfonlayer}
\end{tikzpicture}}
\endpgfgraphicnamed 
  \]
 
 \end{proof}
 
  \begin{lemma}\label{zwsoundness}
 The non-anyonic qudit ZW-calculus is sound, i.e., all the ZW rules as given in Figure \ref{quditrules1},  \ref{quditrules2},  \ref{quditrules3}    still  hold under the standard interpretation  $ \left\llbracket \cdot \right\rrbracket$.
  \end{lemma}
  
  This can be proved by direct calculation or the translation $ \left\llbracket \cdot \right\rrbracket_{WX}$ with the soundness of qudit ZX-calculus. We leave the details to the interested readers. 
Universality  for qudit ZW/ZX-calculus means any matrix $A$ of size $d^m \times d^n$ can be represented by a ZW/ZX  diagram $D$ such that $ \left\llbracket D \right\rrbracket = A$. Since qudit ZX-calculus is universal \cite{wangqufinite} and all the ZX generators can be translated to ZW diagrams via  $ \left\llbracket \cdot \right\rrbracket_{XW}$ without changing the standard interpretation, we have

\begin{theorem}
 The non-anyonic qudit ZW-calculus is universal.
 \end{theorem}

\section{Qufinite ZW-calculus as unified qudit ZW-calculi}
Similar to the qudit ZX-calculus \cite{wangqufinite}, we can also have a unified framework of all qudit ZW-calcului called qufinite ZW-calculus. 

The way to realise the idea of qufinite ZW-calculus is the same as that for qufinite ZX-calculus: label each wire or node with its dimension and add two new generators called dimension-splitter and dimension-binder respectively  \cite{wangqufinite}.  As usual,  a wire labelled with 1 will be depicted as empty. Now we give the generators for qufinite ZW-calculus as follows.

\begin{table}[!h]
\begin{center} 
\begin{tabular}{|r@{~}r@{~}c@{~}c|r@{~}r@{~}c@{~}c|}
\hline
&& & %
\beginpgfgraphicnamed{TikZit//whitespiderditd}
\InputIfFileExists{TikZit//whitespiderditd.tikz}{}{\input{./figures/TikZit//whitespiderditd.tikz}}
\endpgfgraphicnamed  & &&& %
\beginpgfgraphicnamed{TikZit//blacktriangleditd}
\InputIfFileExists{TikZit//blacktriangleditd.tikz}{}{\input{./figures/TikZit//blacktriangleditd.tikz}}
\endpgfgraphicnamed\\\hline
&& & %
\beginpgfgraphicnamed{TikZit//braidld}
\InputIfFileExists{TikZit//braidld.tikz}{}{\input{./figures/TikZit//braidld.tikz}}
\endpgfgraphicnamed  & &&& %
\beginpgfgraphicnamed{TikZit//braidinvld}
\InputIfFileExists{TikZit//braidinvld.tikz}{}{\input{./figures/TikZit//braidinvld.tikz}}
\endpgfgraphicnamed\\\hline
&& & %
\beginpgfgraphicnamed{TikZit//idqudit}
\InputIfFileExists{TikZit//idqudit.tikz}{}{\input{./figures/TikZit//idqudit.tikz}}
\endpgfgraphicnamed  & &&& %
\beginpgfgraphicnamed{TikZit//swapd}
\InputIfFileExists{TikZit//swapd.tikz}{}{\input{./figures/TikZit//swapd.tikz}}
\endpgfgraphicnamed\\\hline
&& & %
\beginpgfgraphicnamed{TikZit//capdit}
\InputIfFileExists{TikZit//capdit.tikz}{}{\input{./figures/TikZit//capdit.tikz}}
\endpgfgraphicnamed  & &&& %
\beginpgfgraphicnamed{TikZit//cupdit}
\InputIfFileExists{TikZit//cupdit.tikz}{}{\input{./figures/TikZit//cupdit.tikz}}
\endpgfgraphicnamed\\\hline
&& & %
\beginpgfgraphicnamed{TikZit//binderdit}
\InputIfFileExists{TikZit//binderdit.tikz}{}{\input{./figures/TikZit//binderdit.tikz}}
\endpgfgraphicnamed  & &&& %
\beginpgfgraphicnamed{TikZit//binderditflip}
\InputIfFileExists{TikZit//binderditflip.tikz}{}{\input{./figures/TikZit//binderditflip.tikz}}
\endpgfgraphicnamed\\\hline
\end{tabular} \caption{Generators of qufinite ZX-calculus, where  the two diagrams at the bottom of the table of generators are called  dimension-binder and  dimension-splitter respectively,  $d, m,n\in \mathbb N, d\geq 2; \protect\overrightarrow{\alpha_d}=(a_1,\cdots, a_{d-1}); a_i\in \mathbb C; i \in \{1,\cdots, d-1\}; j \in \{0, 1,\cdots, d-1\}; s, t \in \mathbb N \backslash\{0\}$. }\label{qbzxgeneratordit}
\end{center}
\end{table}

\FloatBarrier

The rules of qufinite ZW-calculus are of two kinds: those involving  dimension-splitter and dimension-binder  and the others which are qudit ZW rules for any given dimension $d$. Below we only give the rules  including  dimension-splitter and dimension-binder.  
 \begin{figure}[!h]
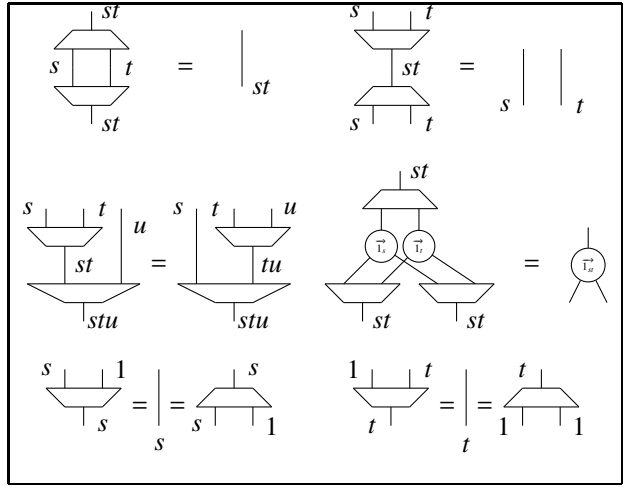

\begin{center} 
\[
\quad \qquad\begin{array}{|cc|}
\hline
\beginpgfgraphicnamed{TikZit//binderunitary1}
\InputIfFileExists{TikZit//binderunitary1.tikz}{}{\input{./figures/TikZit//binderunitary1.tikz}}
\endpgfgraphicnamed&%
\beginpgfgraphicnamed{TikZit//binderunitary2}
\InputIfFileExists{TikZit//binderunitary2.tikz}{}{\input{./figures/TikZit//binderunitary2.tikz}}
\endpgfgraphicnamed\\
    &\\ 
\beginpgfgraphicnamed{TikZit//binderassoc}
\InputIfFileExists{TikZit//binderassoc.tikz}{}{\input{./figures/TikZit//binderassoc.tikz}}
\endpgfgraphicnamed&%
\beginpgfgraphicnamed{TikZit//binderwspider}
\InputIfFileExists{TikZit//binderwspider.tikz}{}{\input{./figures/TikZit//binderwspider.tikz}}
\endpgfgraphicnamed\\
    &\\ 
\beginpgfgraphicnamed{TikZit//binderwith1rt}
\InputIfFileExists{TikZit//binderwith1rt.tikz}{}{\input{./figures/TikZit//binderwith1rt.tikz}}
\endpgfgraphicnamed&%
\beginpgfgraphicnamed{TikZit//binderwith1lt}
\InputIfFileExists{TikZit//binderwith1lt.tikz}{}{\input{./figures/TikZit//binderwith1lt.tikz}}
\endpgfgraphicnamed\\
    &\\ 
  		  		\hline  
  		\end{array}\]      
  	\end{center}
  	\caption{Qufinite ZW-calculus rules involving dimension-splitter and dimension-binder,  where $\protect\overrightarrow{1}_m=\protect\overbrace{(1,\cdots,1)}^{m-1}, m, s, t, u \in \mathbb N \backslash\{0\}.$}\label{qufinitezwrules2}
  \end{figure}  
  
 \FloatBarrier

Now we give the standard interpretation for the new generators of qufinite ZW-calculus.
\[
 \left\llbracket%
\beginpgfgraphicnamed{TikZit//binderdit}
\InputIfFileExists{TikZit//binderdit.tikz}{}{\input{./figures/TikZit//binderdit.tikz}}
\endpgfgraphicnamed\right\rrbracket= \sum_{k=0}^{s-1}\sum_{l=0}^{t-1}\ket{kt+l}\bra{kl}, 
 \quad \quad
  \left\llbracket%
\beginpgfgraphicnamed{TikZit//binderditflip}
\InputIfFileExists{TikZit//binderditflip.tikz}{}{\input{./figures/TikZit//binderditflip.tikz}}
\endpgfgraphicnamed\right\rrbracket= \sum_{k=0}^{st-1}\ket{[\frac{k}{t}]}\ket{k-t[\frac{k}{t}]}\bra{k}, \quad\quad   \left\llbracket%
\beginpgfgraphicnamed{TikZit//swapd}
\InputIfFileExists{TikZit//swapd.tikz}{}{\input{./figures/TikZit//swapd.tikz}}
\endpgfgraphicnamed\right\rrbracket=\sum_{k=0}^{s-1}\sum_{l=0}^{t-1}\ket{kl}\bra{lk},
 \]

\[  \llbracket D_1\otimes D_2  \rrbracket =  \llbracket D_1  \rrbracket \otimes  \llbracket  D_2  \rrbracket, \quad 
 \llbracket D_1\circ D_2  \rrbracket =  \llbracket D_1  \rrbracket \circ  \llbracket  D_2  \rrbracket,
  \]
where 
$s, t \in \mathbb N \backslash\{0\},   \bra{i} =\overbrace{(\underbrace{0,\cdots,1}_{i+1}, \cdots, 0)}^{d}, ~ \ket{i}=(\overbrace{(\underbrace{0,\cdots,1}_{i+1}, \cdots, 0)}^{d})^T,  i \in \{0, 1,\cdots, d-1\},$  and $[r]$ is the integer part of a real number $r$.

Note that in the 1-dimensional Hilbert space $H_1=\mathbb C$, we make the convention that $\ket{0}=1$.

Again by the universality of  qufinite ZX-calculus \cite{wangqufinite} and the fact that all the ZX generators can be translated to ZW diagrams via  $ \left\llbracket \cdot \right\rrbracket_{XW}$ without changing the standard interpretation, we have


\begin{theorem}
 The  qufinite ZW-calculus is universal.
 \end{theorem}

\begin{remark}
If we remove the crossing $\tau$ and its inverse $\tau^{-1}$ from the set of generators, and the corresponding rewriting rules which involving the crossings, then we get a general ZW-calculus which works for arbitrary commutative semirings.   
\end{remark}

\section{ Future work}
In this paper, we propose a non-anyionic qudit ZW-calculus for arbitrary dimension $d$ and unify them into a single framework called qufinite ZW-calculus. By establishing a bilateral transition between ZW-calculus and ZX-calculus we obtain the universality of qudit/qufinite ZW-calculus which allows us to encode any matrix into ZW diagrams.

One obvious thing to do next is to prove the completeness of qudit ZW-calculus. This is expected to be achieved based on a proof of completeness  for qudit ZX-calculus and the translation between the two calculi. 

Another interesting thing would be applying this  non-anyionic qudit ZW-calculus to the problem of qudit entanglement classification.

Finally, we want to mention that qubit ZW-calculus has been  shown as a useful framework for conscious experience and cognition \cite{wcb2021}. So it is exciting to see how this general framework of ZW-calculus could be applied  to a new field  beyond quantum. 

  \section*{Acknowledgements} 

The author would like to acknowledge the grant FQXi-RFP-CPW-2018.  The author also thanks Konstantinos Mei for his helpful comments.

\bibliographystyle{eptcs}
\bibliography{generic} 

\end{document}